\documentclass[11pt]{article}

\usepackage[a4paper, margin=1in]{geometry}

\newif\ifdraft
\newif\ifhideproofs
\drafttrue

\newcommand{\CONGEST}{\ensuremath{\mathsf{CONGEST}}\xspace}
\newcommand{\CLIQUE}{\ensuremath{\mathsf{CONGESTED~CLIQUE}}\xspace}
\newcommand{\LOCAL}{\ensuremath{\mathsf{LOCAL}}\xspace}

\usepackage[utf8]{inputenc}
\usepackage[T1]{fontenc}
\usepackage{graphicx}
\usepackage{grffile}
\usepackage{longtable}
\usepackage{wrapfig}
\usepackage{rotating}
\usepackage[normalem]{ulem}
\usepackage{amsmath}
\usepackage{textcomp}
\usepackage{amssymb}
\usepackage{capt-of}
\usepackage{hyperref}
\usepackage{amsmath}
\usepackage{amssymb}
\usepackage{amsthm}
\usepackage{verbatim}
\usepackage{thm-restate}
\usepackage[capitalize,nameinlink]{cleveref} 

\newtheorem{theorem}{Theorem}
\newtheorem{lemma}[theorem]{Lemma}
\newtheorem{corollary}[theorem]{Corollary}
 \newtheorem{definition}[theorem]{Definition}
\newtheorem{claim}{Claim}

\newcommand{\NN}{\mathbb{N}}

\newcommand{\eps}{\varepsilon}
\newcommand{\poly}{\operatorname{\text{{\rm poly}}}}

\newcommand{\ceil}[1]{\lceil #1 \rceil}

\newcommand{\set}[1]{\left\{#1\right\}}

\DeclareMathOperator{\polylog}{\poly\log}

\usepackage[linesnumbered]{algorithm2e}
\usepackage[noend]{algpseudocode}
\usepackage[textsize=tiny]{todonotes}

\newcommand{\svp}[1]{\ifdraft \todo[color=blue!20]{#1}\fi}

\newenvironment{theorem-repeat}[1]{\begin{trivlist}
		\item[\hspace{\labelsep}{\bf\noindent Theorem \ref{#1} .}]\em }
	{\end{trivlist}}

\begin{document}

\begin{flushleft}

\vspace*{0.8cm}
{\huge\bf Distributed Approximation on Power Graphs}
\vspace{1.0cm}
\end{flushleft}

\newcommand{\auth}[3]{\textbf{#1}$\,\,\,\cdot\,\,\,$#2$\,\,\,\cdot\,\,\,$#3\par\medskip}

\auth{Reuven Bar-Yehuda}
{Technion}
{reuven@cs.technion.ac.il}
\auth{Keren Censor-Hillel}
{Technion}
{ckeren@cs.technion.ac.il}
\auth{Yannic Maus}
{Technion}
{yannic.maus@cs.technion.ac.il}
\auth{Shreyas Pai}
{The University of Iowa}
{shreyas-pai@uiowa.edu}
\auth{Sriram V.~Pemmaraju}
{The University of Iowa}
{sririam-pemmaraju@uiowa.edu}

\begin{abstract}
We investigate graph problems in the following setting: we are given a graph $G$ and
we are required to solve a problem on $G^2$.
While we focus mostly on exploring this theme in the distributed \CONGEST model, 
we show new results and surprising connections to the centralized model of computation.
In the \CONGEST model, it is natural to expect that problems on $G^2$ would be quite difficult to solve efficiently
on $G$, due to congestion.
However, we show that the picture is both more complicated and more interesting.

Specifically, we encounter two phenomena acting in opposing directions:
(i) \textit{slowdown due to congestion} and (ii) \textit{speedup due to structural properties of $G^2$.}

We demonstrate these two phenomena via two fundamental graph problems, namely, 
\textit{Minimum Vertex Cover (MVC)} and \textit{Minimum Dominating Set (MDS)}. Among our many contributions, the highlights are the following.
\begin{enumerate}
\item
In the \CONGEST model, we show an $O(n/\epsilon)$-round $(1+\epsilon)$-approximation algorithm for MVC on $G^2$, while no $o(n^2)$-round
algorithm is known for any better-than-2 approximation for MVC on $G$. 

\item
We show a centralized polynomial time $5/3$-approximation algorithm for MVC on $G^2$, 
whereas a better-than-2 approximation is UGC-hard for $G$.

\item
In contrast, for MDS, in the \CONGEST model, we show an $\tilde{\Omega}(n^2)$ lower bound for a constant approximation factor for MDS on $G^2$, 
		whereas an $\Omega(n^2)$ lower bound for MDS on $G$ is known only for exact computation.
\end{enumerate}
In addition to these highlighted results, we prove a number of other results in the 
distributed \CONGEST model including
an $\tilde{\Omega}(n^2)$ lower bound for computing an exact solution to MVC on $G^2$,
a conditional hardness result for obtaining a $(1+\epsilon)$-approximation to MVC 
on $G^2$, and an $O(\log \Delta)$-approximation to the MDS problem on $G^2$ 
in $\mbox{poly}\log n$ rounds.
Our lower bound reductions also lead to hardness results in the centralized setting.
Specifically, we show that there is no FPTAS for MVC on $G^2$ unless $P = NP$ and there is no 
$(1-\epsilon)\ln n$-approximation for MDS on $G^2$ unless $NP \subseteq
DTIME(n^{O(\log\log n)})$.
\end{abstract}

\clearpage

\tableofcontents
\clearpage

\section{Introduction}
\label{sec:introduction}
The theme of this paper is designing algorithms and proving hardness results for graph problems
in the following setting: \textit{we are given a graph $G$ and we are required to solve the problem 
on the square $G^2$.}
Computing on the square $G^2$ of a communication network $G$ is a crucial primitive in distributed applications, a prime example being the computation of a network decomposition of $G^2$ to obtain derandomization results~\cite{GHK18}. Another example is the problem of coloring $G^2$, which arises in frequency assignment in radio networks~\cite{Fotakis1999,SHAO200837}.

In the \LOCAL model, where message sizes are not bounded, 
computing on $G^2$ incurs just a constant-factor overhead in the complexity of an algorithm. 
Yet, this is far from being true when messages are of bounded size, e.g., consider the problem in which each node needs to learn the input values 
of all of its neighbors in $G^2$ in  the \CONGEST model: since a message only contains $O(\log{n})$ bits, a simple information-theoretic argument gives that the runtime dramatically suffers from \emph{congestion} and the worst case requires a multiplicative 
overhead proportional to the maximum degree of $G$, which is not present if vertices solve the same problem on $G$ instead of $G^2$.

While the above shows that with limited message sizes, computing on $G^2$ potentially suffers from more congestion than computing in $G$, one notices that the graph $G^2$ has more 
\emph{structure} compared to $G$, which could potentially be exploited when solving problems on graphs. A notable example is that $G^2$ contains many cliques -- for each neighborhood of nodes with degree greater than 1 in $G$. These two properties act in opposite directions, and the contributions of this paper are to analyze their effect on two fundamental problems, 
namely, \textit{minimum vertex cover (MVC)} and \textit{minimum dominating set (MDS)}.

Formally, when we say that we solve a problem $\Pi$ on $G^2$, we mean that the input graph is 
$G$ and the output is a solution for $\Pi$ on the graph $G^2=(V,F)$, where $F$ is the set of edges $\{u, v\}$ for which $u$ and $v$ are at most two hops from
each other in $G$. We use $G^2$-$\Pi$ to denote the variants of the problems on $G^2$, e.g.,  
$G^2$-MVC denotes the minimum vertex cover problem with input $G$ that is required to output a minimum size vertex cover 
of $G^2$.

\subsection{Our contributions}
\label{subsection:contributions}
In a nutshell, our main findings are that MVC becomes easier on $G^2$ due to its structure, while for MDS the obstacle of congestion is more substantial. The highlights of our contributions are:
\begin{enumerate}
\item
\label{item:congest-vc-alg}
A deterministic $(1+\eps)$-approximation algorithm for $G^2$-MVC in the \CONGEST\ model\footnote{In the \CONGEST\ model  (\cite{peleg00}) a communication network is abstracted as an $n$-node graph. In synchronous rounds each node can send a $O(\log n)$ bit message to each of its neighbors. The \emph{complexity} is the number of rounds until each node has computed its output, e.g., whether it belongs to a VC or not.} 
which completes in $O(n/\eps)$ rounds, for any $\eps > 0$. We also provide an algorithm with these guarantees for weighted $G^2$-MVC (denoted by $G^2$-MWVC). 

In comparison, for MVC on $G$ in the \CONGEST\ model, 
the fastest algorithm for any better-than-2 approximation factor is the naive $O(n^2)$-round algorithm. 

\item 
\label{item:cent-vc-alg}
A deterministic polynomial time centralized algorithm that gives a $5/3$-approximation for $G^2$-MVC while we also show that MVC remains hard on $G^2$, i.e.,  it does not admit a FPTAS unless $P = NP$. 

Our algorithm should be contrasted with the celebrated 
UGC hardness of a polynomial time algorithm for any better-than-2 approximation~\cite{KhotRJCSS2008}. 
Given the hardness of finding a better-than-2-approximation algorithm for MVC, there is a long line of research on the approximability of MVC on 
specific graph classes \cite{BakerJACM1994,HalldorssonSODA1995,ChenISAAC2000}.
Our result contributes to this line of research.

\item 
\label{item:congest-mds-lb}
A lower bound of $\tilde{\Omega}(n^2)$ rounds for any $c$-approximation algorithm with $c<7/6$ for weighted $G^2$-MDS (denoted $G^2$-MWDS) in the \CONGEST\ model. We also provide such a quadratic lower bound when $c<9/8$ for the unweighted case. 

In comparison, for MDS on $G$ in the \CONGEST\ model, no super-polylogarithmic lower bound is known for any approximation factor that is smaller than $\ln \Delta$, where $\Delta$ is the maximum degree in $G$, and the best algorithm for any better-than-$O(\log\Delta)$ approximation factor is the naive $O(n^2)$-round algorithm. The best known lower bounds are a $\tilde{\Omega}(n^2)$ lower bound for exact MDS, and a $\tilde{\Omega}(n)$ lower bound for a $O(1)$-approximation~\cite{BachrachCDELP19}.
We also point out that Bachrach et al.~\cite{BachrachCDELP19} do consider the $G^2$-MWDS problem and provide lower bounds, e.g., a linear lower bound of $O(\log n)$-approximation.
\end{enumerate}

We stress the contrast between the $O(n/\eps)$ rounds algorithm for $(1+\eps)$-approximation of $G^2$-MVC 
and the $\Omega(n^2)$ lower bound for computing an $9/8$-approximation for $G^2$-MDS. 
The main takeaway from our results is how differently MVC and MDS behave as we go from $G$ to $G^2$ in the \CONGEST\ model, exemplifying the two conflicting properties of computing on $G^2$: congestion and structure. Moreover, we show the following.

\paragraph{Distributed MVC} 
We combine our ideas for $G^2$-MVC approximation in \CONGEST with a randomized voting scheme, to obtain an $O(\log{n}+1/\eps)$ round algorithm for $(1+\eps)$-approximation (for any $\eps>0$) of $G^2$-MVC in the \CLIQUE model\footnote{The \CLIQUE model is similar to the \CONGEST model, but vertices  can send $O(\log{n})$-bits messages to all other nodes, not only to its neighbors in the input graph $G$ \cite{Peleg03}.}.
On the lower bound side, we show that exact $G^2$-MVC requires $\tilde{\Omega}(n^2)$ \CONGEST rounds, corresponding to the same lower bound for MVC on $G$~\cite{Censor-HillelKP17}.
Furthermore, we show that if one could improve our running time in item (\ref{item:congest-vc-alg}) above to be $o(n^{1/2}/\eps)$ rounds, for any $\eps > 0$, then one could obtain any constant-factor approximation algorithm for MVC on $G$ in $o(n^2)$ rounds, which would be a major breakthrough. For example, an $(1+\eps)$-approximation of $G^2$-MVC in $O(n^{1/3}/\eps)$ rounds for every $\eps>0$ would yield a $3/2$-approximation algorithm for $G$-MVC in $O(n^{16/9})$ rounds.

\paragraph{Distributed MDS} Using a randomized 2-neighborhood size estimation technique, we show how to simulate the $O(\log \Delta)$-approximation algorithm for
MDS in $G$~\cite{Censor-HillelD18} to work for $G^2$-MDS with only a constant-factor slow down. This yields a $O(\poly(\log n))$-round, $O(\log \Delta)$-approximation algorithm for solving MDS on $G^2$ in \CONGEST.

\paragraph{Centralized MDS}
In the centralized setting, we provide polynomial time reductions between MDS and $G^2$-MDS.  Together with the NP-completeness proof for MDS of Feige  \cite{FiegeJACM1998} these imply that $G^2$-MDS is NP-complete. Our centralized reductions are approximation factor presevering such that also the result on hardness of centralized approximation of MDS carries over to $G^2$-MDS: If a polynomial-time algorithm can solve $G^2$-MDS  with an 
approximation factor of $(1 - \eps) \ln n$, then $NP \subseteq DTIME\left(n^{O(\log \log n)}\right)$.

\subsection{Technical Challenges} 
\label{subsection:challenges}

We overcome significant technical challenges in obtaining our results. We highlight some of these here, and the rest are discussed in the corresponding sections.
\begin{itemize}
\item[(i)] The benefit of the structure of $G^2$ is that it contains many cliques, e.g., each neighborhood in $G$ induces a clique in $G^2$.
One approach to getting a good approximation to MVC is to repeatedly add disjoint cliques to the vertex cover.
Adding a clique of size $s$ costs our algorithm $s$, but OPT needs to pay $s-1$. However, large cliques
could also be exclusively induced by edges in $G^2 - G$ and these are not easy to find. Our algorithm relies on
a structural property we show: we only need to find cliques induced by neighborhoods in $G$. We show that once such cliques
are found and removed, then the remaining graph becomes sparse enough for fast processing.

\item[(ii)] For obtaining our centralized $5/3$-approximation for MVC on $G^2$, we also rely on the structure of $G^2$. We use the local ratio approach~\cite{Bar-YehudaE83} to take care of small vertex-disjoint parts of the graph, for which an optimal solution has to pay not much less than our cost. An example is taking triangles, for which an optimal solution needs to pay 2 while we pay 3. We also use techniques of finding local maximal matchings, as first done by Gavril as explained in~\cite{GareyJ79}.
Many local-ratio algorithms then take the worst approximation factor among all these parts and that is the approximation factor of the entire solution. However, in our algorithm, after handling some parts of the graph we still remain with a part for which we find a 2-approximation. Still, we avoid paying this in the final approximation factor, by constructing a useful partition, in which the \emph{size of one part} is bounded by a constant fraction of the \emph{optimal solution for another part}. This allows us to take a sloppy approximation for the former, and rely on the latter in order to argue that this still gives a good total approximation factor, rather than taking the worst factor among the parts.

\item[(iii)] For the lower bound constructions in this paper, we use the Alice-Bob lower bound framework developed in \cite{PelegR00} 
for obtaining quadratic and near-quadratic lower bounds for graph problems.
To leverage the current lower bound construction for $G$ and use it for $G^2$, we replace each edge $\{u, v\}$
by a \emph{path gadget} that creates a 2-path between $u$ and $v$ and possibly adds an additional $O(1)$ vertices
Note that this provides the edge $\{u, v\}$ in the square of this graph.
Doing this suffices for the centralized setting, but in the distributed setting it introduces a factor-$n$ \emph{blowup}
in the number of vertices, and no longer provides a lower bound that is quadratic in the number of vertices.
An even bigger challenge is the need to create a constant-factor gap in the output of the lower bound reduction, which necessitates many new techniques in our construction.

\end{itemize}

\subsection{Further Related Work}
There is a vast body of research on approximating MVC and MDS in the sequential
setting; we refer to the references in \cite{Vaziranibook, WilliamsonShmoysbook}
as an introduction to this literature.
Since there has been no progress on approximating MVC with an 
approximation factor that is smaller than $2$, researchers have studied the problem of approximating 
MVC on restricted graph classes, such as planar graphs \cite{BakerJACM1994},
bounded-degree graphs \cite{HalldorssonSODA1995}, and
graphs with perfect matchings \cite{ChenISAAC2000}.
More recently, research in \textit{distributed} approximation algorithms for 
MVC and MDS received a lot of attention.
Bar-Yehuda et al.~\cite{Bar-YehudaCSJACM2017} present a deterministic
$(2+\eps)$-approximation algorithm for MWVC in $O(\log \Delta /\eps \log \log \Delta)$ rounds, where 
the $\Delta$-dependency is optimal due to the lower bound of Kuhn et al.~\cite{KMW16}.
Ben-Basat et al.~\cite{Ben-BasatEKS18} shave off the $\eps$ term in the approximation factor
and present a deterministic 2-approximation algorithm for the MWVC problem in the \CONGEST
model at the cost of an increased runtime of $O(\log n\log \Delta/\log^2\log \Delta)$ rounds.
For MDS, Censor-Hillel and Dory \cite{Censor-HillelD18} obtain an $O(\log \Delta)$-approximation
in $O(\log n\log \Delta)$ rounds, improving on the result by Jia et al.~\cite{JiaRS02} that achieves
this approximation in expectation.
Both of these algorithms are randomized, whereas Ghaffari et al. \cite{GK18} and Deurer et al.~\cite{DKM19}
present deterministic algorithms for MDS with approximation factors $O(\log \Delta^2)$ and $(1+\eps)(1+\ln(\Delta+1))$, respectively. Both algorithms rely on a network decomposition of $G^2$ and \cite{GK18} provided a $2^{O(\sqrt{\log n\log\log n})}$ round CONGEST algorithm to compute such a decomposition; due to the faster decomposition algorithm from \cite{RG19} (that also works for $G^2$ and in the CONGEST model) both algorithms now run  in $\polylog n$ rounds.

On the lower bound side, Kuhn et al.~\cite{KMW16} provide lower bounds of the form
$\Omega(\log \Delta /\log \log \Delta)$ and 
$\Omega(\sqrt{\log n /\log \log n})$ for constant-approximation to MVC 
and polylogarithmic-approximation to MDS in the \LOCAL\ model.
Naturally, these lower bounds apply to the \CONGEST\ model as well.
In the \LOCAL model MVC, MDS, $G^2$-MVC and $G^2$-MDS can be solved deterministically in $\poly\log n$ rounds even if one aims for $(1+\eps)$-approximations with $\eps=1/\poly\log n$ \cite{SLOCAL17,RG19}.

\subsection{Outline}
\begin{itemize}
\item \Cref{sec:notation} formally introduces the problems that we study. 
\item \Cref{sec:DistVCUpperBound} presents our distributed upper bounds for $G^2$-vertex cover. 
\item \Cref{sec:cent-alg} presents the centralized upper bound for $G^2$-MVC.
\item \Cref{app:distVC} presents near-quadratic lower bound for the exact solution of $G^2$-MVC, conditional lower bounds (conditioned on the hardness of $G$-MVC) and limitations of our current lower bound techniques.  
\item \Cref{sec:MDSupper} presents a distributed algorithm for $G^2$-MDS.
\item  \Cref{sec:MDSLB} presents our lower bounds for approximating $G^2$-MDS (weighted an unweighted). 
\item \Cref{sec:centHardness} shows that there is no FPTAS for $G^2$-MVC or $G^2$-MDS unless $P=NP$. 
\end{itemize}

\section{Problems \& Notation}\label{sec:definitions}
\label{sec:notation}

A \emph{vertex cover (VC)} of a graph $G=(V,E)$ is a subset $S\subseteq V$ of the vertices such that for any edge $\{u,v\}\in E$ at least one of its endpoints is contained in $S$ and $|S|$ is its \emph{size}. 
A \emph{dominating set (DS)} of a graph $G=(V,E)$ is a subset $S\subseteq V$ of the vertices such that any vertex $v\in V$ is in $S$ or has a $G$-neighbor in $S$ and $|S|$ is its \emph{size}. In the \emph{minimum vertex cover problem (MVC)} or the \emph{minimum dominating set problem (MDS)} the objective is to compute a VC (DS) of minimal size among all feasible VCs (MDs). An $\alpha$-approximation to the MVC (MDS) problem is a VC (DS) $S$ with $|S|/|OPT|\leq \alpha$ where $OPT$ is a solution with minimal size for the problem in the respective graph. When solving the VC or DS problem on $G^2=(V,F)$  (with input graph $G=(V,E)$) the solution is a subset $S\subseteq V$ of the vertices for which all other feasibility notions are interpreted with regard to the edge set $F$.
We also consider the weighted versions of these problems, \textit{minimum weighted vertex cover (MWVC)} and \textit{minimum weighted dominating set (MWDS)}.

When we solve $G^2$-MVC or $G^2$-MDS problems in the \CONGEST or \CLIQUE model we require that at the end of the algorithm each node needs to know whether it is part of the vertex cover or the dominating set. We point out that nodes cannot decide locally (see \cite{FKP13}) whether a given given vertex cover (or DS) has a 'good' approximation factor as it might approximate an optimal solution badly in some part of the graph while it still provides a good approximation on the whole graph. 

\paragraph{Notation.} For some subset $S\subseteq V$ $G^2[S]$ denotes the subgraph of $G^2$ induced by the vertex set $S$, that is, it contains an edge between any two vertices $u,v\in S$ if and only if $u$ and $v$ have distance at most two in $G$; we explicitly point out that the distance is measured in $G$.
For a vertex $v\in V$ we denote its non-inclusive neighborhood in $G$ by $N(v)$.

 \section{Distributed $G^2$-Minimum Vertex Cover (Algorithms)} 
\label{sec:DistVCUpperBound}
In this section we show our distributed upper bounds for $G^2$-MVC. In \Cref{sssec:UBunweightedVC} we present an $O(n)$ algorithm to compute a $(1+\eps)$ approximation. In \Cref{sssec:UBweightedMVC} we extend this bound for the weighted version of the problem. In \Cref{sssec:UBcongestedClique} we show that one can compute a $(1+\eps)$-approximation for $G^2$-MVC in $O(\log n)$ rounds in the \CLIQUE. In \Cref{app:distVC} we prove lower bounds for distributed $G^2$-vertex cover and show limitations of the current lower bound techniques.
\subsection{\CONGEST: $(1+\eps)$-Approximation for $G^2$-MVC}
\label{sssec:UBunweightedVC}
This section is devoted to proving the following theorem.
\begin{theorem}\label{thm:G2VC}
For any $\eps>0$ there is a deterministic distributed \CONGEST\ algorithm that computes a $(1+\eps)$-approximation of $G^2$-minimum vertex cover with communication network $G$ in $O(n/\eps)$ rounds.
\end{theorem}

We prove Theorem~\ref{thm:G2VC}. We first explain the algorithm and then prove its correctness, approximation factor and runtime.

\IncMargin{0.5em}
\RestyleAlgo{boxruled}
\begin{algorithm} \caption{The $(1 + \eps)$-approximation for $G^2$-Minimum Vertex Cover} \label{alg:VCepsSimplified}
$C = V$ \tcp*{possible centers}
$R = V$ \tcp*{vertices not in the cover}
$S = \emptyset$ \tcp*{Vertices in the Vertex Cover}
\While{there is a node $c \in C$ with $|N(c) \cap R| > 1/\eps$}{\label{alg:VCeps:outerbegin}
  Remove $c$ from $C$ \\
  Add $N(c)$ to $S$ \tcp*{same as adding $N(c)\cap R$} 
  Remove $N(c)$ from $R$ \tcp*{same as removing $N(c)\cap R$}
}\label{alg:VCeps:outerend}
$U=V\setminus S$\\
Elect a leader $\ell\in V$ and let it learn the following set of edges
$F=\{\{u,v\}\in E \mid u\in U, v\in V\}$\\
Leader $\ell$ computes an optimal solution $R^*$ of the VC problem on $H=G^2[U]$ using $F$ \label{alg:line:bruteforce}\\
\Return $S \cup R^*$
\end{algorithm}
\DecMargin{0.5em}

\medskip

\noindent\textbf{Algorithm:}
The algorithm consists of two phases (see Algorithm \ref{alg:VCepsSimplified} for Pseudocode). In the first phase we carefully and iteratively add vertices  to an initially empty set $S$ such that, 
(1) $S$ is a good approximation for all edges of $G^2$ that it covers, and 
(2) the graph $H=G^2[V\setminus S]$ of all edges of $G^2$ that are not covered by $S$ can be efficiently learned by a leader vertex $\ell$ in the second phase. 
 Then, the leader $\ell$ computes an optimal vertex cover $R^*$ of $H$ and we return the union of $S$ and $R^*$. The approach of repeatedly covering disjoint parts of the graph for which one can prove a good approximation compared with any optimal solution is common for computing MVCs, and is used throughout our algorithms (also in other sections).

\textbf{Phase I:}  We continue with explaining how the first phase can be executed in a sequential manner. At the start of the phase the \emph{'cover'} $S$ is empty, we denote the set of \emph{remaining nodes} by $R=V\setminus S$, i.e., the nodes that have not yet been added to the cover $S$, and $C=V$ denotes the set of possible \emph{centers}. Then, as long as there is a center node $c\in C$ that has more than $1/\eps$ neighbors in $R$---the neighbor relation is the neighbor relation in the communication graph $G$ and not the one in the square graph $G^2$---node $c$ adds all of its neighbors to $S$ and removes them from $R$. Then $c$ leaves the set $C$. 
As the runtime of the second phase will dominate anyhow there is no need to efficiently parallelize this sequential algorithm. Instead we use an arbitrary symmetry breaking between vertices in $c$ with the help of their ID to run the sequential algorithm in a distributed manner: Any vertex with degree at least $1/\eps$ in $R$ is a \emph{candidate} and any candidate who has the maximum ID in its two hop neighborhood adds its neighbors to $S$, removes them from $R$ and leaves $C$.

\textbf{Phase II:} We continue with explaining how the leader $\ell$ learns the graph $H$. After phase I let $U=V\setminus S$ be the vertices that are not in the cover yet and define the following set of edges
\begin{align}
F=\{\{u,v\}\in E \mid u\in U, v\in V\}~.
\end{align}
We show that the leader $\ell$ can learn the set $F$ efficiently (\Cref{lem:learnRemainingGraph}) and that it can compute the set $H$ with the knowledge of $F$(\Cref{lem:computeRemainingGraph}). We point out that the graph $H$ can be have a large number of edges and we do not explicitly send all its edges to the leader vertex $\ell$ but instead $\ell$ only learns the much smaller set of edges $F$ and then uses it to locally compute $H$. 

Note that all steps of the algorithm except for line \ref{alg:line:bruteforce} (which is executed locally inside one vertex) only reason about $G$ and in particular no single condition or action refers to $G^2$. All reasoning about $G^2$ (except for  line \ref{alg:line:bruteforce}) is only part of the analysis.

\begin{lemma}[Learning Remaining Graph]\label[lemma]{lem:learnRemainingGraph}
The leader vertex $\ell$ can learn the set $F$ in $O(n/\eps)$ rounds.
\end{lemma}
\begin{proof}
Consider the setting in which each node of the communication graph has at most $c$ distinct pieces of information. 
By building a BFS tree with a leader as the root and pipelining messages the leader vertex can learn the pieces of information in $O(c\cdot n)$ rounds. 

Any node in $v\in V$ has at most $1/\eps$ neighbors in $U$ as otherwise $v$ would be processed in the first phase and all of its neighbors would join $S$. We make $v$ responsible for sending its $1/\eps$ incident edges of $F$ to the leader. Using the aforementioned pipelining argument leader $\ell$ learns the set $F$ in $O(n/\eps)$ rounds.
\end{proof}
\begin{lemma}
\label[lemma]{lem:computeRemainingGraph}
The graph $H=G^2[U]$ can be computed using the knowledge of $F$.
\end{lemma}
\begin{proof}
We use the edge set $F$ to form the following graph $H'=(U,F')$ with
\begin{align}
F' & =F \cup F_1'\text{ where}\\
F'_1 & =\{\{u,v\}\mid u,v\in U\text{, exists $w$ with }\{u,w\}, \{v,w\}\in F, \} 
\end{align}
Let $H=G^2[U]=(U,E_H)$. We show that $H=H'$, i.e., that $F'$ equals $E_H$.
First, let $e=\{u_1,u_2\}\in E_H$. If $e\in E(G)$, then $e\in F\subseteq F'$ by the definition of $F$. If $e\notin E$, then there exists a $w \in V$ with $e_1=\{u_1,w\}\in E$ and $e_2=\{w,u_2\}\in E$. As both edges $e_1$ and $e_2$ have at least one endpoint in $U$ they are contained in $F$. Thus $e\in F'$ by the definition of $F'_1$.

 For the reverse inclusion let $e=\{u_1,u_2\}\in F'$. By the definition of $F$ and $F_1'$ the edge $e$ is an edge of $G^2[V]$. As both of its endpoints are not in $S$ the edge $e$ is an edge of $G^2[U]$, i.e., $e\in E_H$. 
\end{proof}

\begin{lemma}[Valid Vertex Cover]\label[lemma]{lem:validVC}
The computed set $S\cup R^*$ is a valid vertex cover of $G^2$. 
\end{lemma}
\begin{proof}
Any edge $\{u,v\}$ with at least one of the endpoints in $S$ is covered as all vertices in $S$ are contained in the cover. Thus let $\{u,v\}$ be an edge of $G^2$ with $u,v\in U$. As $R^*$ is a vertex cover of $G^2[U]$ the edge is covered by $R^*$. 
\end{proof}

\begin{lemma}[$S$ approximates well]\label[lemma]{lem:Sapprox} Let $l\in\NN_{>0}$ be an arbitrary positive integer. If the algorithm is executed with $\eps=1/l$ we obtain
$|S|\leq (1+\eps)\cdot |O|$ where $O$ is any vertex cover of the graph $G^2[S]$.
\end{lemma}
\begin{proof}
If no vertex is processed in the loop we have $S=\emptyset$ and the claim holds trivially. Otherwise, 
  let $c_1,\dots,c_k$ be the nodes in $C$ that are chosen in loop, according to their order in which they are processed and for $i=1,\ldots,k$ let $S_i=N(v_i)\cap R$ be the set of vertices that join the set $S$ when node $c_i$ is processed, define $d_i=|S_i|$ and obtain the partition
  $S=S_1\cup \ldots \cup S_k$~.  
	From the condition in the while loop, we have, $d_{i} > 1/\eps$, but since \(d_{i}\) can only be an integer, we get $d_{i} \ge \lfloor 1/\eps \rfloor + 1$, and in particular $d_i\geq l+1\geq 2$. 
  Define $O_i=O\cap S_i$ and obtain the partition $O=O_1\cup \ldots \cup O_k~.$
  For any $i=1,\ldots,k$ the graph $G^2[S_i]$ forms a clique and $O_i$ has to cover all edges with both endpoints in $S_i$. Thus we obtain that $|O_i|\geq |S_i|-1=d_i-1>0$.  We obtain the following calculation in which we never divide by zero due to $d_i-1>0$
    \begin{align}
      \frac{|S|}{|O|} &\le \frac{\sum_{i=1}^{k}{d_i}}{\sum_{i=1}^{k}{(d_i - 1)}}
                      = \frac{\sum_{i=1}^{k}{d_i}}{(\sum_{i=1}^{k}{d_i}) - k}
                      = 1 + \frac{k}{(\sum_{i=1}^{k}{d_i}) - k}\\
											& \leq 1 + \frac{k}{(\sum_{i=1}^{k}{(\lfloor 1/\eps \rfloor+1)}) - k} 
											= 1+ \frac{k}{k\cdot \lfloor 1/\eps \rfloor  + k - k} \\
											&=  1+ \lfloor 1/\eps \rfloor^{-1}
											=1+\eps
    \end{align}

    The last equality follows as  $1/\eps=l$ is an integer and the claim follows by multiplying both sides of the inequality with $|O|$. 
		Note that $O$ is only charged at most once for covering every edge of $G^2[S]$.
\end{proof}

\begin{proof}[Proof of \Cref{thm:G2VC}] 
  If $\eps>1$ simply add all vertices to the cover and obtain a $2$-approximation. Note that a $2$-approximation of $G^2$-MVC is a trivial task that requires no communication. To see why, note that the complemented of any solution $S$ to $G^2$-MVC is an inclusion maximal independent set (MIS) $I$ in $G^2$. The size of any $G^2$-MIS in a connected $n$-node graph $G$ is upper bounded by $n/2$ as one can pair any vertex in $I$ with a distinct vertex in $V\setminus I$ (see \Cref{lem:structAllNodesVC} for a formal proof). Thus, taking all nodes into the cover gives a $2$-approximation.
  Otherwise  $\eps'=1/l$ where $l=\lceil 1/\eps\rceil$ and apply the aforementioned algorithm with $\eps'$ instead of $\eps$---note, if $1/\eps$ is an integer we have $\eps=\eps'$.

\noindent\textbf{Correctness:} The set $S\cup R^*$ is a valid vertex cover due to \Cref{lem:validVC}.

\noindent\textbf{Runtime:} The set $R$ shrinks by at least $1/\eps'$ vertices in each iteration of the loop and one iteration (with an arbitrary symmetry breaking as explained before) can be implemented in $O(1)$ rounds in CONGEST. Thus the number of iterations is upper bounded by $|V|/(1/\eps')=\eps'|V|$ and the first phase can be executed in $O(\eps'\cdot |V|)=O(\eps\cdot n)$ rounds. Learning the sets $F_1$ and $F_2$ takes $O(n/\eps')=O(n/\eps)$ rounds due to \Cref{lem:learnRemainingGraph}.
Computing the optimal solution $R^*$ of $G^2[U]$ is done locally and the solution can be distributed to all nodes in $O(n)$ rounds. Thus the runtime is $O(\eps\cdot n+n/\eps+n)=O(n/\eps)$~. 

\noindent\textbf{Approximation Factor:} We show that the resulting vertex cover $S\cup R^*$ is a $(1+\eps)$-approximation of the VC of $G^2$.
Let $OPT$ be an optimal VC of $G^2$, let $S$ and $U=V\setminus S$ be the sets after the execution of the loop. 
Define $OPT_S=OPT\cap S$ and $OPT_U=OPT\cap U$. 
As $OPT_S$ is a vertex cover of $G^2[S]$  \Cref{lem:Sapprox} implies that $|S|\leq (1+\eps')|OPT_S|$. 
Vertices in $OPT_S$ cannot cover any edge in $G^2[U]$ and thus $OPT_U$ is a vertex cover of $G^2[U]$. As $R^*$ is an optimal vertex cover of $G^2[U]$ we obtain 
 $|R^*|\leq |OPT_U|$. We obtain
\begin{align*}
  |S\cup R^*|&\leq |S|+|R^*|\leq (1+\eps')|OPT_S|+|OPT_U| \\
             &\leq (1+\eps')|OPT|\leq (1+\eps)|OPT|~.  & \qedhere
\end{align*}
\end{proof}

 We begin by proving the fact that every vertex cover of $G^r$ is quite large and this yields a trivial better-than-$2$ approximation for $G^r$, even for relatively small $r$.
Therefore, we have a \(0\)-round approximation algorithm for unweighted vertex cover.
We now formalize this result.
\begin{lemma}
\label{lem:structAllNodesVC}
For a connected \(n\)-vertex graph \(G\), the size of any vertex cover in
\(G^r\) \(1 \le r \le n\) is at least \(n - n/\alpha\) where \(\alpha = \lfloor r/2 \rfloor + 1\).
Thus, a solution that includes all vertices is a $(1 + 1/\lfloor r/2 \rfloor)$-approximation to unweighted MVC in \(G^{r}\).
\end{lemma}
\begin{proof}
Consider any independent set \(I\) in \(G^r\). Two vertices in \(I\) must be at least distance \(r+1\) apart in \(G\) and for every vertex \(u \notin I\), there can be at most one \(v \in I\) such that the distance between \(u\) and \(v\) in \(G\) is at most \(\lfloor r/2 \rfloor\). This is because if there is more than one such vertex then \(I\) is no longer an independent set in \(G^r\). Moreover, this is true for every vertex on the path from \(u\) to \(v\) in \(G\). In other words, for every vertex \(v \in I\), we can assign at least \(\lfloor r/2 \rfloor\) unique vertices that are not in \(I\). This implies \(|I| < n/\alpha\) as otherwise there will be more than \(n\) vertices in \(G\).

Since the complement of any vertex cover is an independent set, the lemma follows.
\end{proof}
The above lemma implies that the solution containing all vertices is a \(2\)-approximation for unweighted vertex cover in \(G^2\) and the approximation factor goes closer to \(1\) as \(r\) is increased.

\subsection{\CONGEST: $(1+\eps)$-Approximation for $G^2$-MWVC}
\label{sssec:UBweightedMVC}

We now show how to extend the algorithm for MVC on input $G^2$ using communication network $G$, described in the previous section,
to minimum \textit{weighted} vertex cover (MWVC). 
Since all 0-weight vertices can be included in the vertex cover with no cost, we assume without loss of generality, that all
vertex weights are positive.
For ease of exposition, we assume that every vertex weight can be represented in $O(\log n)$ bits.

\begin{theorem}
\label{theorem:MWVCUpperBound}
For any $\eps > 0$, there is a deterministic, distributed \CONGEST algorithm that yields a $(1+\eps)$-approximation
in $O(n\log n /\eps)$ rounds for the MWVC problem on input $G^2$ with communication network $G$.
\end{theorem}

We make two changes to Algorithm \ref{alg:VCepsSimplified}.
\begin{itemize}
\item[(i)]  In Algorithm \ref{alg:VCepsSimplified}, we repeatedly picked a vertex $c$ with a large enough still-active neighborhood (i.e., $N(c) \cap R$) 
to ensure that when $N(c) \cap R$ is added to the vertex cover, we continue to get a good approximation.
The cardinality of active neighborhoods is not useful in the weighted setting, but we can derive a corresponding 
condition for picking $c$ as follows.
For any vertex $c$, let $W(c)$ denote $\sum_{v\in N(c) \cap R} w(v)$. Let $w^*(c)$ be the maximum weight of a vertex in
$N(c) \cap R$. Then, $W(c)- w^*$ is a lower bound on the weight of an optimal vertex cover
of $G^2$ restricted to $N(c) \cap R$, i.e., the clique in $G^2$ induced by $N(c) \cap R$. 
Therefore, to be able to safely add $N(c) \cap R$ to the vertex cover the following condition needs to be satisfied:
$W(c) \le (1+\eps) (W(c) - w^*(c))$, or equivalently
\begin{equation}
\label{equation:weightedCondition}
w^*(c) \le W(c) \cdot \frac{\eps}{(1+\eps)}.
\end{equation}
\item[(ii)] 
We apply the above condition, not to the entire active neighborhood of $c$, but to subsets with similar weights.
Let $w_*(c)$ denote the minimum weight of a vertex in $N(c)$.
We partition $N(c)$ into subsets $N_i(c) := \{v \in N(c) \mid w_*(c) \cdot 2^i \le w(v) < w_*(c) \cdot 2^{i+1}\}$
for $i = 0, 1, \ldots, I$, where $I = O(\log_2 n)$ (since all vertex weights have $O(\log n)$-bit representations).
Instead of checking condition (\ref{equation:weightedCondition}) for $N(c) \cap R$, we check it for
$N_i(c) \cap R$, for each $i$.
Let $w^*_i(c)$ denote the maximum weight of a vertex in $N_i(c) \cap R$ and similarly let
$W_i(c)$ denote $\sum_{v\in N_i(c) \cap R} w(v)$
Specifically, we replace Line 4 in Algorithm \ref{alg:VCepsSimplified} by 
$$\textbf{while}\textit{ there is a vertex $c \in C$ and $i$ with } w^*_i(c) \le W_i(c) \cdot \frac{\eps}{(1+\eps)}\textbf{ do}$$
and perform the body of the loop with $N_i(c)$ replacing $N(c)$ in Lines 6 and 7.
\end{itemize}

\noindent
To ensure efficiency of our algorithm, the key property we need is for $|F|$ to be small, so that a leader can gather all of $F$ 
and the algorithm can proceed to Phase II.
In the analysis of Algorithm \ref{alg:VCepsSimplified}, this simply followeed from the fact that after Phase I, every vertex has at most $1/\eps$ neighbors
in $U$ (the set of vertices not in the cover).
The following lemma proves a similar condition for the current algorithm.
\begin{lemma} 
\label{lemma:UpperBoundF}
$|F| = O\left(n \cdot \frac{(1+\eps)}{\eps} \cdot \log_2 n\right)$.
\end{lemma}
\begin{proof}
Suppose that for all $c \in V$ and $i$, we have $w^*_i(c) > W_i(c) \cdot \eps/(1+\eps)$.
This is guaranteed after Phase 1 of the algorithm.
Let $s$ denote the number of vertices in $N_i(c) \cap R$.
Since all vertices in $N_i(c) \cap R$ have weights in the range $[w_*(c) \cdot 2^i, w_*(c) \cdot 2^{i+1})$, we see
that $w_*(c) \cdot 2^{i+1} > w^*_i(c)$ and $W_i(c) \ge s \cdot w_*(c) \cdot 2^i$.
This leads to the inequality 
$$w_*(c) \cdot 2^{i+1} > s \cdot w_*(c) \cdot 2^i \cdot \frac{\eps}{(1 + \eps)},$$ 
which in turn yields the upper bound
$$s < 2 \cdot \frac{1+\eps}{\eps}.$$ 
Thus, each vertex $c$ has at most $O((1+\eps)/\eps \log_2 n)$ neighbors in $U$ after Phase I and the lemma follows
by accounting for all vertices.
\end{proof}
The rest of the running time analysis simply follows as in the corresponding steps 
for Algorithm \ref{alg:VCepsSimplified} (see Lemmas \ref{lem:learnRemainingGraph}, \ref{lem:computeRemainingGraph}). 
Note that in the current algorithm, Phase I runs in $O(n \log n)$ rounds because we sequentially 
consider every vertex $c$ and neighbor set $N_i(c)$ for $O(\log_2 n)$ possible values of $i$.
Phase II run in $O(n \cdot \frac{(1+\eps)}{\eps} \cdot \log_2 n)$ rounds because of the size
of $|F|$.
The correctness follows immediately, as in Lemma \ref{lem:validVC}.
The approximation factor analysis depends on the fact that whenever we add $N_i(c) \cap R$ to the vertex cover,
the weight of the added vertices is within an $(1+\eps)$ factor of what the optimal solution pays
to cover the edges in the subgraph of $G^2$ induced by $N_i(c) \cap R$.
The calculations follow the steps in the proof of Lemma \ref{lem:Sapprox}.

 \subsection{\CLIQUE:  $(1+\eps)$-Approximation for  $G^2$-MVC}
\label{sssec:UBcongestedClique}
In the \CLIQUE we obtain faster deterministic and randomized algorithms for $(1+\eps)$-approximation of MVC on $G^2$. 
As one component these algorithm use that learning the set $F$ is much faster as formalized in the next lemma. 

\begin{lemma}[Learning Remaining Graph in the \CLIQUE]\label{lem:learnRemainingGraphClique}
The leader vertex $\ell$ can learn the set $F$ in $O(1/\eps)$ rounds.
\end{lemma}
\begin{proof}
Just as in the proof of \Cref{lem:learnRemainingGraph} any node in $v\in V$ has at most $1/\eps$ neighbors in $U$ as otherwise $v$ would be processed in the first phase and all of its neighbors would join $S$. We make $v$ responsible for sending its $1/\eps$ incident edges of $F$ to the leader which can be done in parallel for all vertices in $1/\eps$ rounds. 
\end{proof}
\Cref{lem:learnRemainingGraphClique} together with the analysis from \Cref{sssec:UBunweightedVC} immediately implies the following corollary. 
\begin{corollary}[\CLIQUE, deterministic]
For any (also non constant) $\eps>0$ there is a deterministic distributed \CLIQUE algorithm that computes a $(1+\eps)$-approximation to the $G^2$-minimum vertex cover in $O(\eps\cdot n+1/\eps)$ rounds. By setting $\eps=1/\sqrt{n}$ we can compute a $(1+1/\sqrt{n})$-approximation in $O(\sqrt{n})$ rounds, deterministically. 
\end{corollary}
\begin{proof}
Learning the set $F$ is sufficient to compute the graph $H$ and can be done in $O(1/\eps)$ rounds by \Cref{lem:learnRemainingGraphClique}. Distributing the locally computed solution for the graph $H$ can be done in one round. Thus the runtime is dominated by the $O(\eps\cdot n)$ rounds of the first phase.
\end{proof}

We now show that we can also speed up the first phase of our our algorithm in Section~\ref{sssec:UBunweightedVC}, using randomization. 
 This follows a similar approach used in~\cite{Censor-HillelD18} for approximating spanners and MDS, which in turn uses is a modification of the framework of Jia et al.~\cite{JiaRS02} for approximating minimum dominating sets. While this faster implementation itself works in the \CONGEST model it still does not improve the overall running time in the \CONGEST model. However, combined with \Cref{lem:learnRemainingGraphClique} it allows us to obtain a much faster algorithm for the \CLIQUE model, as given in the following theorem.

\begin{theorem}
\label{theorem:VC-CC}
For any $\eps>0$, there is a distributed \CLIQUE algorithm that computes a $(1+\eps)$-approximation for $G^2$-MVC in $O(\log n+1/\eps)$ rounds. 
\end{theorem}

\begin{proof}
We use the same notation of our algorithm in Section~\ref{sssec:UBunweightedVC}: The set $S$ contains vertices in the cover, the set $R$ denotes the remaining vertices, and $C$ is the set of candidates. Whenever the degree of a vertex $c\in C$ in $R$ drops below the threshold $8/\eps+2$ it leaves $C$, that is, any vertex $c \in C$ for which $|N(c)\cap R| \leq 8/\eps+2$ is removed from $C$. For simplicity of notation, denote $|N(c)\cap R|$ by $d_R(c)$.

The algorithm consists of $O(\log n)$ phases and in each phase some vertices leave $C$, some vertices are added to $S$ and removed from $R$. 
In each phase, each candidate, that is, each vertex in $C$, informs its neighbors that it is a candidate. Then, every vertex in $R$ \emph{votes} for one of its candidate neighbors and informs all of them about its vote---we will soon explain the details of the voting scheme. Each candidate $c$ which gets at least $d_R(c)/8$  votes is \emph{successful}, i.e., its neighbors are added into $S$ and are removed from $R$, and the candidate $c$ is removed from $C$. We repeat until there are no more candidates and then resort to having a leader learn the edges in $F$ as in our algorithm in Section~\ref{sssec:UBunweightedVC}. 

We now describe the voting mechanism. Each candidate $c$ chooses a random number $r_c \in [n^4]$ and a voter votes for its candidate neighbor who has the highest random value. A candidate $c$ is \emph{successful} if it gets at least $d_R(c)/8$ votes. 

\paragraph{Correctness:} Lemma~\ref{lem:validVC} holds here too, proving that we cover all edges of $G^2$.

\paragraph{Approximation Factor:} The approximation proved in Lemma~\ref{lem:Sapprox} is maintained because we charge the votes that made the candidate successful only to a single candidate and a candidate is only successful if it got at least $d_R(c)/8\geq (8/\eps+2)/8>1/\eps$  votes.

\paragraph{Runtime:} We analyze progress using a potential function whose value at the beginning of iteration $i$ is $\Phi_i = \sum_{c \in C}{d_R(c)}$. We claim that $\Phi$ decreases by a constant factor in each iteration and hence we have a logarithmic number of iterations until the potential function is smaller than $1$, i.e., the set of candidates $C$ is empty. For each vertex $v \in R$ we denote by $s(v)$ the number of neighbors it has in $C$. For each candidate $c \in C$, we sort its neighbors in $R$ according to their $s$ values, and split them into sets $T(c)$ and $B(c)$ where the $\ceil{d_R(c)/2}$ top values go into $T(c)$ and the $\ceil{d_R(c)/2}$ bottom values go into $B(c)$ (there may be an overlap of one vertex). For every $v \in T(c)$, we say that $(c,v)$ is a \emph{top pair}. We show that if a voter of a top pair votes for a candidate, then the candidate is successful with constant probability.

\begin{claim}
\label{claim:conditional-success}
If $(c,v)$ is a top pair then $Pr[c \text{ is successful} \mid v \text{ votes for } c]\geq 1/3$.
\end{claim}

\begin{proof}
We first claim that if $v_1,v_2 \in N(c)\cap R$ and $s(v_1) \geq s(v_2)$ then $$Pr[v_2 \text{ votes for } c \mid v_1 \text{ votes for } c]\geq 1/2.$$
Let $N_1,N_2$ and $N_{1,2}$ be the number of candidates that are neighbors of $v_1$ but not of $v_2$, the number of candidates that are neighbors of $v_2$ but not of $v_1$, and the number of candidates that are neighbors of both $v_1$ and $v_2$, respectively. Then,
\begin{align*}
Pr[v_2 \text{ votes for } c \mid v_1 \text{ votes for } c] & = \frac{Pr[v_1 \text{ and } v_2 \text{ vote for } c]}{Pr[v_1 \text{ votes for } c]} = \frac{1/(N_1+N_2+N_{1,2})}{1/(N_1+N_{1,2})}  \\ 
& = \frac{N_1+N_{1,2}}{N_1+N_2+N_{1,2}} \geq 1/2, 
\end{align*}
where the last inequality is because $N_1 \geq N_2$, since $s(v_1) \geq s(v_2)$.

Now, let $(c,v)$ be a top pair and suppose that $v$ votes for $c$. Then each $u \in B(c)$ votes for $c$ w.p. at least $1/2$ because $s(v) \geq s(u)$ for each such $u$. Let $x$ be the number of vertices in $B(c)$ that \emph{do not} vote for $c$. We have that $E[x] \leq |B(c)|/2$. By Markov's inequality, $Pr[x > 3 \cdot B(c)|/4] \leq 2/3$. Thus, w.p. at least $1/3$, at least $|B(c)|/4$ vertices in $B(c)$ vote for $c$. This is at least $d_R(c)/8$ and hence $c$ is successful.
\end{proof}

We can now show that in expectation, $\Phi$ decreases by a constant factor in each iteration. A Chernoff bound then implies that we need only $O(\log n)$ iterations w.h.p. We then need to learn the remaining graph after we are done, which, by \Cref{lem:computeRemainingGraph} and a proof along the same lines as the one for \Cref{lem:learnRemainingGraphClique}---note the slightly different threshold of $8\eps+1$  for being a candidate--- can be done in $O(1/\eps)$ rounds.
This gives a total of $O(\log n + 1/\eps)$ rounds in the \CLIQUE model.

Thus, it remains to show that the expected decrease in $\Phi$ is a constant fraction of it. Recall that we define $\Phi_i = \sum_{c \in C}{d_R(c)}$. If we count this according to the vertices in $R$, we get that this equals $\sum_{v \in R}{s(v)}$. If a vertex $v$ votes for a successful candidate $c$, then $\Phi$ decreases by at least $s(v)$. We can associate this decrease with the pair $(c,v)$ because $v$ votes for a single candidate. We then have:
\begin{align*}
E[\Phi_{i-1} - \Phi_i] &\geq \sum_{(c,v)}{Pr[v \text{ votes for } c \text{ and } c \text{ is successful}]\cdot s(v)}\\
&\geq \sum_{\text{top pairs } (c,v)}{Pr[v \text{ votes for } c]\cdot Pr[c \text{ is successful} \mid v \text{ votes for } c] \cdot s(v)}\\
&\geq \sum_{\text{top pairs } (c,v)}{(1/s(v)) \cdot (1/3) \cdot s(v)} = 1/3 \cdot |\{\text{top pairs } (c,v) \} |\\
&\geq 1/6 \cdot \Phi_{i-1},
\end{align*}
where the last inequality follows since at least half of the pairs are top pairs.
\end{proof}

\section{Centralized $G^2$-Minimum Vertex Cover}\label{sec:cent-alg}
In this section we present a polynomial time centralized algorithm that gives an $\alpha$-approximation to unweighted MVC, for a constant $\alpha < 2$. 
In \Cref{sec:centHardness} we show that the problem of computing an exact $G^2$-MVC is NP-hard and that one cannot even get a FPTAS unless $P=NP$. 
Here,  we formally  show the following theorem.

\begin{theorem}\label{theorem:VCcent}
There is a centralized polynomial time algorithm that computes an $\alpha$-approximation to $G^2$-Minimum Vertex Cover, for a constant $\alpha < 2$.
\end{theorem}
\textbf{High Level View on Algorithm:}
The algorithm consists of three parts, in each of which we find an approximate solution to part of the remaining graph, until it is empty. 
The high-level goal is to find a set of nodes $U$ for which the size of an optimal solution $|OPT_U|$ can be well approximated and is larger by some positive fraction compared to the size of the rest of the nodes $U'$, which allows us to find only a sloppy approximation for the cover of $U'$.  

In the algorithm, we maintain that $V'$ and $E'$ are the remaining sets of vertices and edges, respectively. Initially, these are $V$ and $E(G^2)$. We denote by $S$ the cover that we obtain, initially empty. During the algorithm, whenever we say that we take a node into $S$, we also mean that it is removed from $V'$ and all edges with at least one endpoint in $S$ are removed from $E'$. Whenever there is a node with degree $0$, it is removed from $V'$.

In the first part, the algorithm loops until there are no more triangles in $E'$: Sequentially, as long as there is a triangle in $E'$ we add all of its three vertices to $S$, remove the vertices from $V'$ and we remove all edges touching the triangle from $E'$, i.e., we remove all edges with at least one endpoint being one of the triangle's vertices. 
In the second part, that we detail on later, we remove further vertices from $V'$ and edges from $E'$ such that the remaining part has minimum degree $4$. Then, in the third part of the algorithm we compute a $2$-approximation on the remaining vertex cover instance (e.g. by computing a maximal matching and adding both endpoints of the matched edges to the cover). 
For detailed pseudocode we refer to Algorithm \ref{alg:VC}.

\IncMargin{0.5em}
\RestyleAlgo{boxruled}
\begin{algorithm} \caption{An $\alpha$-approximation for Minimum $G^2$-Vertex Cover} \label{alg:VC}
$V' = V$, $E' = E(G^2)$ \tcp*{current set of vertices and edges}
$S = \emptyset$ \tcp*{vertices in the cover}
$V_1 = \emptyset$, $V_2 = \emptyset$, $V_3 = \emptyset$ \tcp*{part-1, part-2, part-3 vertices in the cover}

\DontPrintSemicolon\tcp*{part-1}
\While{there is a triangle in $(V',E')$ }{\label{alg:VC1} 
  Take all three nodes of the triangle into $S$ and into $V_1$, delete them from $V'$  and their incident edges from $E'$\\
}
{}\tcp*{part-2}
\While{there is a node $x \in V'$ with $\deg_{E'}(x) \leq 3$}
{\label{alg:VC2}
  If there is a node $x$ with $\deg_{E'}(x)=1$ then its neighbor is taken into $S$ and into $V_2$\\
  else If there is a node $x$ with $\deg_{E'}(x)=2$, denote its neighbors by $y_1,y_2$. Since there is no node $u$ with $\deg_{E'}(u)=1$, it holds that $y_1$ has a neighbor $z\neq x$. Take $z,y_1,y_2$ into $S$ and into $V_2$\\
  else If there is a node $x$ with $\deg_{E'}(x)=3$, denote its neighbors by $y_1,y_2,y_3$. Since there is no node $u$ with $\deg_{E'}(u)<3$, there are two nodes $z_1\neq z_2$ such that $z_1$ is a neighbor of $y_1$ and $z_2$ is a neighbor of $y_2$, and $z_1, z_2 \neq x,y_1,y_2,y_3$ because there are no triangles. We take $y_1,y_2,y_3,z_1,z_2$ into $S$ and into $V_2$\\

In all three cases nodes added to $S$ are removed from $V'$ and their incident edges are removed from $E'$.
}
\tcp*{part-3}
Find a 2-approximation on $(V',E')$ and take its nodes into $S$ and into $V_3$.\label{alg:VC3}\\

\Return $S$
\end{algorithm}
\DecMargin{0.5em}

Let $V_1$ be the vertices added to $S$ in the first part, let $V_2$ be the set of vertices added to $S$ in the second part and let $V_3$ be the set of vertices added to $S$ in the third part. For, for $i=1,2,3$ let $W_i$ denote the set of vertices that leave $V'$ in phase $i$. The set $W_i$ contains $V_i$ and it also contains all vertices that leave $V'$ because their degree in the remaining graph reached $0$.  
Further, for $i=1,2,3$ we denote $s_i=|V_i|$. 
We first show some crucial properties that hold after the first part. We call every edge in $G$ a \emph{red edge}, and every edge in $G^2-G$ is called a \emph{blue edge}. For a subset  $F$ of edges of $G^2$, we denote by $F^{\mathrm{red}}$ and $F^{\mathrm{blue}}$ the red and blue edges in $F$, respectively.  Let $R=(V',E')=(V_R,E_R)$ be the remaining graph after the first part.
\begin{lemma}\label[lemma]{lemma:VC1}
The following properties hold for $R$.
(1) There are no triangles in $R$.
(2) $E^{\mathrm{red}}_R$ forms a matching.
(3) $s_1 \geq E^{\mathrm{blue}}_R$.

\end{lemma}

\begin{proof}
As an edge of $G^2$ can only be removed from $E'$ if at least one of its endpoints is added to $S$ we obtain that the remaining graph $R$ equals $G^2[V_R]$, that is, it contains all edges induced by vertices in $V_R=V'$.  We need this property to prove all three parts. 

\begin{enumerate}
\item The graph $R$ clearly has no triangles, as otherwise the loop started in Line~\ref{alg:VC1} is not finished. 
\item 
The set $E^{\mathrm{red}}_R$ forms a matching, as otherwise it has two adjacent edges $\{x,y\}, \{y,z\}$ with $x,y,z\in V_R$, but since these are red edges this implies that the edge $\{x,z\}$ is in $E_R=G^2[V_R]$ and so the triangle $\{x,y\}, \{y,z\}, \{z,x\}$ is contained in $R$, which is a contradiction.
\item We claim that for every edge $e=\{x,y\} \in E^{\mathrm{blue}}_R$, there is at least one vertex $v \in V_1$ that forms a triangle with $e$. This is because there must be such a vertex in $V$, and if it is not in $V_1$ it has to still be in $V_R=V\setminus V_1$ but this implies the triangle $\{x,y,z\}$ in $R$. Moreover, it holds that if $v$ forms a triangle with $e=\{x,y\} \in E^{\mathrm{blue}}_R$ then it does not form a triangle with any other edge $e' \in E^{\mathrm{blue}}_R$: Assume it did and denote one of the endpoints of $e'$ that is different from $x$ and $y$ by $x'$. Then $x,y$ and $x'$ are neighbors of $v$ in $G$ and $x,y,x'\in V_R$ from which we can deduce that the edges $e$, $\{x,x'\}$ and $\{x',y\}$ are all part of $R$ and form a triangle, a contradiction. 
Therefore, we have that $s_1 \geq E^{\mathrm{blue}}_R$. \qedhere
\end{enumerate}
\end{proof}

We also keep track of the approximation factor we have so far. Let $OPT_1$ be an optimal cover for the edges of $G^2$ induced by $W_1$ and let $opt_1 = |OPT_1|$.
The triangles taken into $V_1$ are vertex disjoint and $OPT_1$ must take at least two nodes of every triangle while we add all three vertices of the triangle. This implies $opt_1 \geq (2/3) s_1$.
Next, we show that the properties for $R$ also hold after part 2. Let $R'=(V',E')=(V_{R'},E_{R'})$ be the remaining graph after Line~\ref{alg:VC2}. 

\begin{lemma}\label[lemma]{lemma:VC2}
The following properties hold for $R'$.
(1) There are no triangles in $R'$.
(2) $E^{\mathrm{red}}_{R'}$ forms a matching.
(3) $s_1 \geq E^{\mathrm{blue}}_{R'}$.
(4) $s_1 \geq (3/2)|V_{R'}|$.

\end{lemma}

\begin{proof}
First note that we again have that $R'$ equals $G^2[V_{R'}]$ as an edge of $G^2$ can only be removed from $E'$ if at least one of its endpoints is added to $S$.
Since we only remove nodes and edges, the graph $R'$ retains the three properties of $R$ from Lemma~\ref{lemma:VC1} because these are monotone properties. That is, $R'$ has no triangles, $E^{\mathrm{red}}_{R'}$ forms a matching, and $s_1 \geq E^{\mathrm{blue}}_{R'}$, where $E^{\mathrm{blue}}_{R'}$ is the set of blue edges in $R'$. 

The graph $R'$ has the additional property that for every $v \in R'$ it holds that $\deg_{E_{R'}}(v) \geq 4$, and since $E^{\mathrm{red}}_{R'}$  forms a matching, it holds that for every $v \in R'$ we have $\deg_{E^{\mathrm{blue}}_{R'}}(v) \geq 3$.  This gives that  $s_1 \geq E^{\mathrm{blue}}_{R'} \geq (1/2)\sum_{v\in R'}{\deg_{E^{\mathrm{blue}}_{R'}}(v)} \geq (1/2) \cdot 3|V_{R'}| = (3/2)|V_{R'}|$.
\end{proof}

We also keep track of the approximation factor we have so far. Let $OPT_2$ be an optimal cover for the edges of $G^2$ induced by $W_2$ and let $opt_2 = |OPT_2|$.
For a node $x$ with $\deg_{E'}(x)=1$ there is a single vertex taken into $V_2$ and $OPT_2$ must also take a vertex to cover that edge. For a node $x$ with $\deg_{E'}(x)=2$ there are 3 vertices taken into $V_2$ and $OPT_2$ must take at least 2 nodes to cover the vertex-disjoint edges $\{z,y_1\},\{x,y_2\}$. Finally, for a node $x$ with $\deg_{E'}(x)=3$ there are 5 nodes taken to $V_2$, and $OPT_2$ must take 3 nodes to cover the vertex-disjoint edges $\{y_1,z_1\},\{y_2,z_2\},\{x,y_2\}$. The latter dominates the ratio, giving that $opt_2 \geq (3/5) s_2$

Let $OPT_3$ be an optimal cover for the edges of $G^2$ induced by $W_3$ (these are the edges $E_{R'}$) and let $opt_3 = |OPT_3|$. Let $s_3 = |V_3|$. Because $V_3$ is a 2-approximation for $R'$, we immediately get that $opt_3 \geq (1/2) s_3$. These three inequalities are encapsulated in the following lemma.

\begin{lemma}\label[lemma]{lemma:OPTs}
It holds that $opt_1 \geq (2/3) s_1$, $opt_2 \geq (3/5) s_2$, and $opt_3 \geq (1/2) s_3$.
\end{lemma}

We next show that the computed set actually is a vertex cover.
\begin{lemma}\label[lemma]{lem:SvertexCover}
At the end of the algorithm the set $S=V_1\cup V_2\cup V_3$ is a vertex cover of $G^2$.
\end{lemma}
\begin{proof} We only remove an edge from $E'$ if at least one of its endpoints is added to $S$. 
Thus any edge that is not covered by a vertex in $V_1$ or $V_2$ is still contained in $E'$ after phase 2 and any such edge is then covered by $V_3$.
\end{proof}

We are now ready to prove Theorem~\ref{theorem:VCcent}.

\begin{proof}(Proof of Theorem~\ref{theorem:VCcent})
\Cref{lem:SvertexCover} shows that the returned set $S=V_1\cup V_2\cup V_3$ is a vertex cover and the runtime of the algorithm is polynomial. 
We now bound the approximation factor. Let $OPT$ be an optimal solution for $G^{2}$, and let $opt = |OPT|$ and let $s = |S|$. It holds that $opt \geq opt_1 + opt_2 + opt_3$ because these are optimal solutions for the vertex-disjoint sets of edges $W_1$, $W_2$ and $W_3$. Since $s=s_1 + s_2 + s_3$, the approximation factor is $\alpha = s/opt \leq (s_1 + s_2 + s_3)/(opt_1 + opt_2 + opt_3)$. We bound the value of $\alpha$ as follows:

From Lemma~\ref{lemma:VC2}, we know that  $s_1 \geq (3/2)|V_{R'}| \geq (3/2)s_3$. Denoting $c = s_1/s_3$, we have that $c \geq 3/2$. We now have $s = s_1 + s_2 + s_3 = s_2 + (c+1)s_3$. By Lemma~\ref{lemma:OPTs} and using $s_1=c\cdot s_3$, we have
\begin{align*}
  opt &\geq  opt_1 + opt_2 + opt_3 \geq (2/3) s_1 + (3/5) s_2 + (1/2) s_3 \\
      &=  (3/5) s_2 + ((2/3)\cdot c+(1/2))s_3 = (3/5)s_2 + ((4c+3)/6))s_3~.
\end{align*}

We claim that $(4c+3)/6 \geq (3/5)(c+1)$, and thus $opt\geq (3/5)s_2 + (3/5)(c+1)s_3 =  (3/5) (s_2 + (c+1)s_3) = (3/5)s$, which proves that $\alpha \leq 5/3$.
For $(4c+3)/6 \geq (3/5)(c+1)$ to hold, we need $5(4c+3) \geq 6\cdot3(c+1)$, that is, we need $20c +15 \geq 18c+18$, which is equivalent to $2c \geq 3$. This holds since our bound for $c$ is exactly $c \geq 3/2$, by Lemma~\ref{lemma:VC2}.
\end{proof}

By plugging in the result of \Cref{theorem:VCcent} in the second phase of the algorithm of \Cref{thm:G2VC} or \Cref{theorem:VC-CC} we obtain the following corollary.
\begin{corollary}There exists a deterministic 
\CONGEST\ algorithm in which nodes only use polynomial computations and that computes a $5/3$-approximation for $G^2$-MVC in $O(n)$ rounds.
There exists a randomized 
\CLIQUE\ algorithm in which nodes only use polynomial computations and that computes a $5/3$-approximation for $G^2$-MVC in $O(\log n)$ rounds.

\end{corollary}
\begin{proof}
We run the first phase of the algorithm from \Cref{thm:G2VC} (\CONGEST) or \Cref{theorem:VC-CC} (\CLIQUE) with $\eps=1/2$. Then, in the second phase we learn the remaining graph at a leader vertex which locally uses \Cref{theorem:VCcent} to compute a $5/3$-approximation for the remaining graph. The approximation factor is the maximum of $1+\eps=3/2$ and $5/3$.
\end{proof}

\section{Distributed $G^2$-Minimum Vertex Cover (Lower Bounds)}\label{app:distVC}

In \Cref{subsec:aloceBob} we present the \emph{Alice-Bob lower bound framework} developed in \cite{PelegR00} that we use to obtain quadratic and near-quadratic lower bounds for $G^2$-MVC and $G^2$-MDS (Section \ref{sec:MDSLB}). Then, in Sections \ref{sec:org8d0f7a0} and \ref{sec:orgcb69afc} we use the framework to prove near-quadratic lower bounds for $G^2$-MWVC and $G^2$-MVC. In \Cref{sec:VC-apx-limitations} we show limitations of the framework and in \Cref{sec:VChardness} we show our conditional lower bound for $G^2$-MVC.
\subsection{Reduction from Communication Complexity: The Alice Bob Framework}
\label{subsec:aloceBob}
To prove our lower bounds
we use the known framework of reductions from 2-party communication problems or reduce to lower bounds that have been proven with this framework. This framework was pioneered by Peleg and Rubinovich~\cite{PelegR00}, and has been used extensively since then to obtain lower bounds for bandwidth restricted models (see, e.g.,~\cite{AbboudCHK16,Dassarmaetal12,CzumajK18,FischerGKO18,FrischknechtHW12,Censor-HillelKP18,PanduranganPS16,Elkin04}).
The novelty in our proofs lies in the constructions of the graph families that give our reductions. We first recall the framework itself, as follows.

The 2-party communication setting consists of two players, Alice and Bob, who are given two input strings, $x,y\in\set{0,1}^K$ respectively,
and need to evaluate some given function  $f:\set{0,1}^K\times\set{0,1}^K\to\set{\mathrm{true},\mathrm{false}}$ on their inputs.
The maximal number of bits, over all inputs, exchanged in a protocol $\pi$ that computes
$f$ is the \textit{communication complexity of $\pi$} and is denoted $CC(\pi)$. The \textit{communication complexity of the function $f$} is the minimum of $CC(\pi)$ over all deterministic protocols $\pi$ that compute $f$ and is denoted $CC(f)$. In a similar manner, for randomized protocols, the randomized communication complexity of $f$ is denoted $CC^R(f)$.
In the \textit{set disjointness} problem, the problem is to compute a boolean function 
$DISJ_K$ defined as $DISJ_K(x, y) = \mathrm{false}$ if and only if there is an index $0\leq i\leq K-1$ such that $x_i=y_i=1$. It is well known that that $CC(DISJ_K), CC^R(DISJ_K)$ are both $\Theta(K)$ (see, e.g.,~\cite{KushilevitzN:book96}).

\begin{definition} (Family of Lower Bound Graphs~\cite{Censor-HillelKP17})\label[definition]{def:family}
Fix an integer $K$, a function $f:\set{0,1}^K\times\set{0,1}^K\to\set{\mathrm{true},\mathrm{false}}$ and a graph predicate $P$. A family of graphs
\[ \left\{G_{x,y}=(V,E_{x,y}) \mid x,y \in \set{0,1}^K\right\}\] with a partition $V=V_A\dot\cup V_B$ is said to be a family of \emph{lower bound graphs for the \CONGEST\ model w.r.t. $f$ and $P$} if the following properties hold:
\begin{enumerate}
  \item\label{ItemInLBGraphs: va}
	  Only the existence or the weight of edges in $V_A\times V_A$ may depend on $x$;
  \item\label{ItemInLBGraphs: vb}
	  Only the existence or the weight of edges in $V_B\times V_B$ may depend on $y$;
  \item\label{ItemInLBGraphs: pandf}$G_{x,y}$ satisfies the predicate $P$ iff $f(x,y)=\mathrm{true}$.
\end{enumerate}
\end{definition}

\begin{theorem}(\cite{Censor-HillelKP17})\label[theorem]{thm: general lb framework}
Fix a function $f:\set{0,1}^K\times\set{0,1}^K\to\set{\mathrm{true},\mathrm{false}}$ and a graph predicate $P$. If there is a family $\{G_{x,y}\}$ of lower bound graphs for the \CONGEST\ model w.r.t.~$f$ and $P$ with cut edge set $C = E(V_A, V_B)$, then any deterministic algorithm for deciding $P$ in the \CONGEST\ model requires $\Omega (CC(f)/|C|\log n)$ rounds, and any randomized algorithm for deciding $P$ in the \CONGEST\ model requires $\Omega (CC^R(f)/|C|\log n)$ rounds.
\end{theorem}

\subsection{\CONGEST: Quadratic Lower Bound for Exact $G^2$-MWVC (Warmup)}
\label{sec:org8d0f7a0}
We show an \(\Omega(n^2)\) lower bound for exact computation of $G^2$-MWVC. We later make it apply also for $G^2$-MVC.

\begin{theorem}\label[theorem]{thm:MWVC-lower-bound-main-theorem}
  Any distributed algorithm in the \CONGEST\ model which given an input graph \(G\), computes the minimum weighted vertex cover of \(G^{2}\) requires $\tilde{\Omega}(n^{2})$ rounds.
\end{theorem}

To prove this lower bound, we use the framework of reductions from 2-party communication problems as stated in Theorem~\ref{thm: general lb framework}.

A naive attempt is to try and use the vertex cover lower bound graph \(G_{x,y}\) from~\cite{Censor-HillelKP17} and replace each edge with a vertex of weight \(0\), 
in order to get a new lower bound graph \(H_{x,y}\) with the same size of solution for $H^2$ as in $G$. 
The issue is that in the graph \(H_{x,y}\), the number of vertices is \(O(m)\) where \(m = O(n^{2})\) is the number of edges in \(G_{x,y}\). This is a major issue if we want super linear lower bounds as a quadratic lower bound in \(G_{x,y}\) only gives a linear lower bound for \(H^{2}_{x,y}\). So instead, we modify the lower bound graph construction of~\cite{Censor-HillelKP17} in a subtle manner to show an \(\Omega(n^2)\) lower bound for computing exact MWVC in \(G^2\).

\textbf{The \(G\)-MVC lower bound graph family \(G_{x,y}\) from~\cite{Censor-HillelKP17}: } See~\Cref{fig:censor-hillel-mvc-lb-graph} for an illustration of \(G_{x,y}\). At a high level, the lower bound graph of~\cite{Censor-HillelKP17} has four cliques \(A_1, A_2, B_1, B_2\) of size \(k\) which are called the \textit{row vertices} and \(2 \log_2 k\) \(4\)-cycles which are called \textit{bit-gadgets}. There are \(\log_{2} k\) \(4\)-cycles for the row vertices \(A_{1}, B_{1}\) and the other \(\log_{2} k\) \(4\)-cycles are for the row vertices \(A_{2}, B_{2}\)
The \(i^{th}\) bit gadget for \(A_{1}, B_{1}\) is a \(4\)-cycle with vertices \(t_{A_{1}}^{i}, f_{A_{1}}^{i}, t_{B_{1}}^{i}, f_{B_{1}}^{i}\). The vertices in \(a_{1}^{i} \in A_{1}\) are connected to the bit gadget vertices \(f_{A_{1}}^j, t_{A_{1}}^j\) depending on the binary representation of \(i-1\).
Specifically, $a_1^i$ is connected to $t_{A_{1}}^{j}$ if the $j^{th}$ bit of the binary representation of $i-1$ is $1$ and it is connected to $f_{A_{1}}^{j}$ otherwise.
For example, the vertex \(a_{1}^{1}\) is connected to all the \(f_{A_{1}}\) vertices. The connections for other row vertices in  \(A_{2}, B_{1}, B_{2}\) to the corresponding bit gadget vertices are similar.

An edge between vertices \(a_{1}^{i} \in A_1\) and \(a_{2}^{j} \in A_2\) is added iff \(x_{ij} = 0\) in the set disjointness input \(x \in {\{0,1\}}^{k^{2}}\). Similarly, an edge between vertices \(b_{1}^{i} \in B_1\) and \(b_{2}^{j} \in B_2\) is added iff \(y_{ij} = 0\) in the set disjointness input \(y \in {\{0,1\}}^{k^{2}}\).

\begin{figure}
\begin{center}
\includegraphics[width=.9\linewidth]{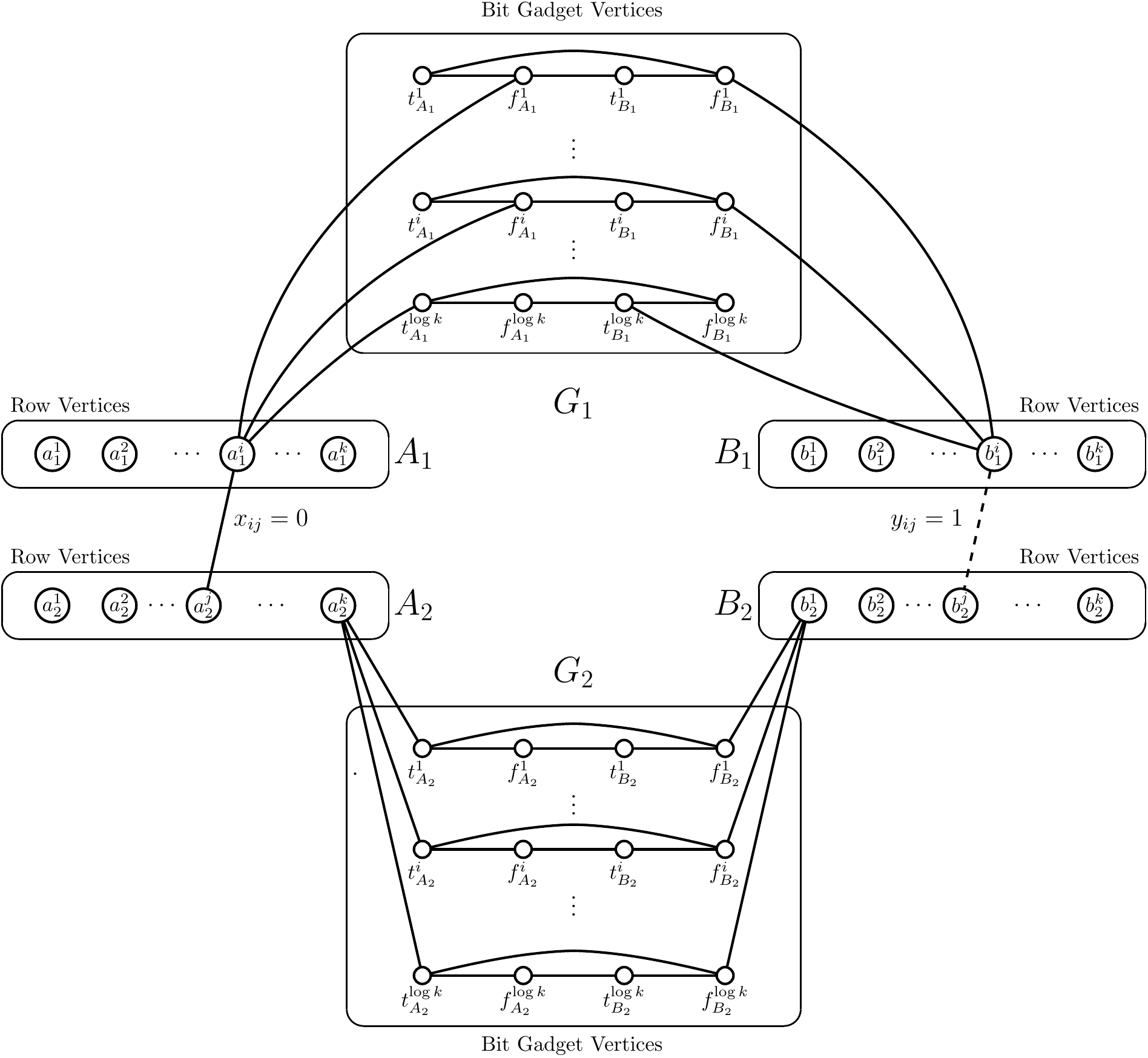}
\end{center}
\caption{\label[figure]{fig:censor-hillel-mvc-lb-graph}This figure shows the lower bound graph \(G_{x,y}\) that appears in~\cite{Censor-HillelKP17} used to show a quadratic lower bound for computing exact MVC in the \CONGEST\ model. We use this graph as the basis for our vertex cover lower bounds.}
\end{figure}

\textbf{\(G^{2}\)-MWVC lower bound graph family \(H_{x,y}\): } In order to construct our lower bound graph \(H_{x,y}\), we first start with \(G_{x,y}\). We take the edges incident on the \(2 \log_2 k\) bit-gadget vertices and replace each edge \(e\) with a \emph{path gadget} \(P_{e}\) which is a single vertex \(p_{e}\) having weight \(0\), which connected to both endpoints of \(e\).
Note that up to this point, we have added \(O(k \log k)\) vertices of weight \(0\), so we do not have too many vertices in \(H_{x,y}\). But we cannot replace each edge between the cliques \(A_1, A_2\) and cliques \(B_1, B_2\) by path gadgets because doing so might introduce \(O(k^2)\) vertices in the worst case.

To overcome this issue, we have the cliques \emph{share} their path gadgets. In particular, we connect a new zero weight vertex \(p_{a}^{i}\) to the vertex \(a_{1}^{i} \in A_1\) for each \(1 \le i \le k\) and for every edge between \(a_1^{i} \in A_1\) and \(a_2^{j} \in A_2\) in \(G_{x,y}\), we add a corresponding edge between \(p_{a}^{i}\) and \(a_{2}^{j}\).  We do the same for the row vertices \(B_{1}, B_{2}\) by connecting a new zero weight vertex \(p_{b}^{i}\) to the vertex \(b_{1}^{i} \in B_{1}\) for each \(1 \le i \le k\). Therefore, the number of vertices in \(H_{x,y}\) is still \(O(k \log k)\). See \Cref{fig:MWVC-lower-bound-path-gadgets} for an illustration. Note that all vertices in \(H_{x,y}\) that come from \(G_{x,y}\) have weight 1 and all the other vertices have weight \(0\).

\begin{figure}
\begin{center}
\includegraphics[width=.83\linewidth]{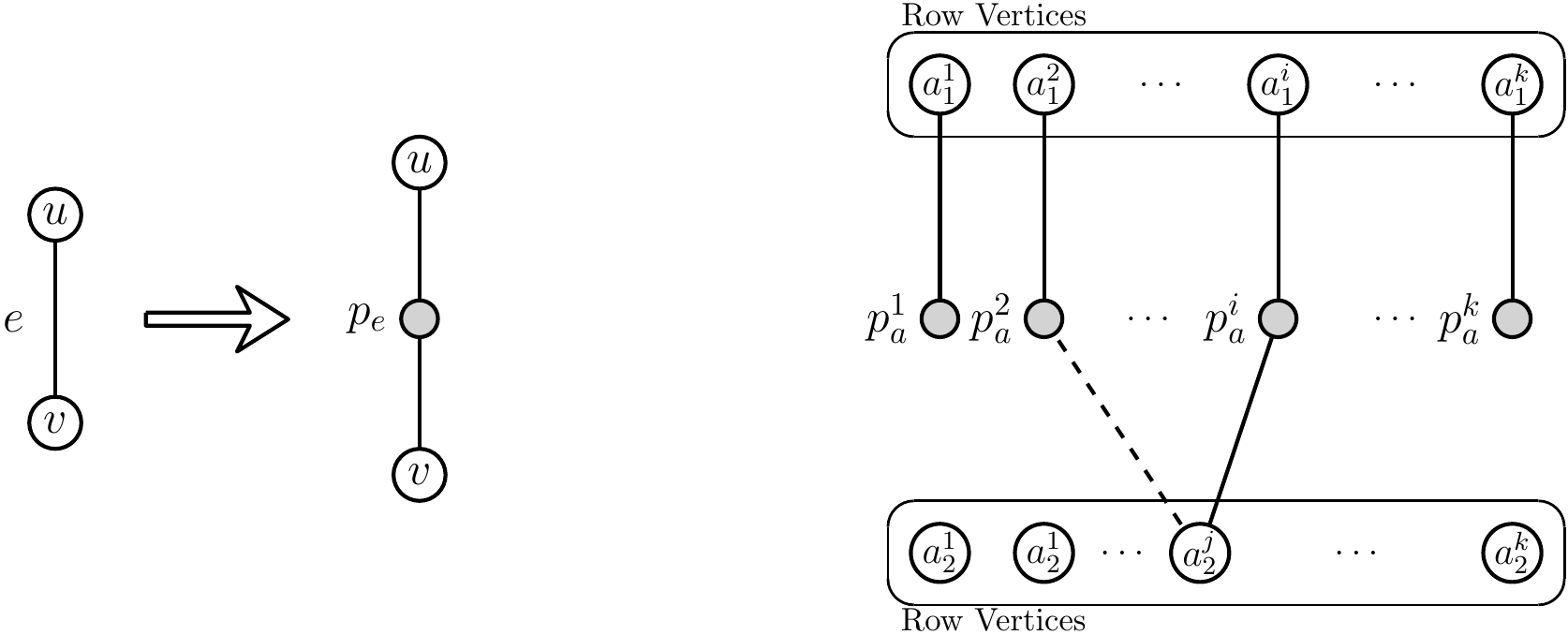}
\end{center}
\caption{\label[figure]{fig:MWVC-lower-bound-path-gadgets}This figure shows how an edge \(e\) is replaced by a path gadget \(P_{e}\) on the left, and on the right it shows how the row vertices in \(A_{1}\) and \(A_{2}\) share their path gadgets. We show two examples of how we add edges between \(p_{a}^{i}\) and \(a_{2}^{j}\) depending on whether the edge \(\{a_{1}^{i}, a_{2}^{j}\}\) exists in \(G_{x,y}\) or not. The gadget sharing for \(B_{1}\) and \(B_{2}\) is similar. We only show some edges of \(H_{x,y}\) for clarity.}
\end{figure}

We formally state our reduction between the lower bound graphs in the following lemma.

\begin{lemma}\label[lemma]{lemma:sizeMWVC}
The graph \(H_{x,y}^2\) has a vertex cover of weight \(W\) if and only if the graph \(G_{x,y}\) has a vertex cover of weight \(W\).
\end{lemma}
\begin{proof}
For the forward direction, we can include all the zero weight vertices in any vertex cover of \(H_{x,y}^2\) without affecting the weight. The edges that need to be covered are between pairs of vertices that are connected by some path gadget, plus the edges of the cliques \(A_{1}, A_{2}, B_{1}, B_{2}\). These are exactly the edges in \(G_{x,y}\) and therefore, the non-zero weight vertices in a vertex cover of \(H_{x,y}^2\) form a valid vertex cover of \(G_{x,y}\). For the reverse direction, notice that a vertex cover of \(G_{x,y}\) along with all zero-weight vertices of \(H_{x,y}\) covers all the edges in \(H_{x,y}^2\).
\end{proof}

\begin{proof}[Proof of \Cref{thm:MWVC-lower-bound-main-theorem}]
  Censor-Hillel et al.~\cite{Censor-HillelKP17} show that the MVC lower bound graph \(G_{x,y}\) is a family of lower bound graphs for the \CONGEST\ model wrt the set-disjointness function \(f=DISJ_{k^{2}}\) and the predicate \(P_{G}\) which asks whether the graph \(G\) has a vertex cover of size \(W = 4(k - 1) + 4\log k\). The vertices of \(G_{x,y}\) are partitioned into \(V_{A} = A_{1} \cup A_{2} \cup \{f_{S}^{i}, t_{S}^{i} \mid 1 \le i \le \log_{2} k, S \in \{A_{1}, A_{2}\}\}\) and \(V_{B} = V(G_{x,y}) \setminus V_{B}\), with the cut size being \(|E(V_{A}, V_{B})| = O(\log k)\).

  By \Cref{lemma:sizeMWVC}, we know that \(G_{x,y}\) satisfies the predicate \(P_{G}\) iff the graph \(H_{x,y}\) satisfies the predicate \(P_{H}\) which asks whether the input graph has a weighted \(G^{2}\)-vertex cover of weight \(W\). Recall that the number of vertices in \(H_{x,y}\) is \(O(k \log k)\).

  Define \(V_{A}' = V_{A} \cup \{P_{e} | e = \{u, v\} \text{ and } u, v \in V_{A}\}\cup \{p_{a}^{i} \mid 1 \le i \le k\}\) and \(V_{B}' = V(H_{x,y}) \setminus V_{A}'\). With these definitions of \(V_{A}'\) and \(V_{B}'\), the size of the cut of \(H_{x,y}\) is \(|E(V_{A}', V_{B}')| = O(\log k)\).

  The graph \(H_{x,y}\) with partition \(V_{A}', V_{B}'\) is a family of lower bound graphs wrt the set-disjointness function \(f=DISJ_{k^{2}}\) and the predicate \(P_{H}\). Therefore, \Cref{thm: general lb framework} gives an \(\tilde{\Omega}(k^2)\) lower bound for the problem of deciding whether a graph with \(O(k \log k)\) vertices has a vertex cover of weight \(W = 4(k - 1) + 4 \log_{2} k\) as shown in~\cite{Censor-HillelKP17}. Therefore for a graph with \(n\)-vertices we get an \(\tilde{\Omega}(n^{2})\) lower bound which completes the proof of \Cref{thm:MWVC-lower-bound-main-theorem}.
\end{proof}

\subsection{\CONGEST : Quadratic Lower Bound for Exact $G^2$-MVC}\label{sec:orgcb69afc}
In this section we show a quadratic lower bound for exact minimum vertex cover with no weights. 

\begin{theorem}\label[theorem]{thm:MVC-lower-bound-main-theorem}
  Any distributed algorithm in the \CONGEST\ model which given an input graph \(G\), computes the minimum vertex cover of \(G^{2}\) requires \(\tilde{\Omega}(n^{2})\) rounds.
\end{theorem}

\textbf{The Dangling Path Gadget: } The lower bound graph construction is inspired by the weighted construction, but here we need to define a new gadget in order to remove the vertex weights. For an edge \(e \in G_{x,y}\), let \(DP_{e}\) be a gadget having \(3\) vertices \(DP_{e}[1], DP_{e}[2], DP_{e}[3]\) connected in a path. The vertex \(DP_{e}[1]\) is connected to both the endpoints of \(e\). We call \(DP_{e}\) a \emph{dangling path gadget} and the path \(DP_{e}[1], DP_{e}[2], DP_{e}[3]\) is called a \emph{dangling path}. We  refer to \(DP_{e}[3]\) as the \emph{leaf} of the dangling path gadget \(DP_{e}\). See~\Cref{fig:MVC-lower-bound-dangling-path-gadget} (left) for an illustration.

\textbf{\(G^{2}\)-MVC lower bound graph family \(H_{x,y}\): } In order to construct our lower bound graph \(H_{x,y}\), we first start with \(G_{x,y}\). We take each edge \(e \in G_{x,y}\) that is incident on the \(2 \log_2 k\) bit-gadget vertices and replace each edge \(e\) with a dangling path gadget \(DP_{e}\). Note that there are $O(k \log k)$ edges incident on bit gadget vertices and therefore we have not introduced too many path gadget vertices.

The remaining edges are those between the row vertices \(A_{1}, B_{1}, A_{2}, B_{2}\).
These can be $O(k^2)$ in number and so we cannot add a dangling path gadget for each edge.
For each row vertex \(a_{1}^{i} \in A_{1}\), we add a \textit{shared path gadget} \(A_{1}^{i}\).
The gadget is similar to the dangling path gadget in that it has \(3\) vertices \(A_{1}^{i}[1], A_{1}^{i}[2], A_{1}^{i}[3]\) connected to form a path.
The vertex \(A_{1}^{i}[1]\) is connected to \(a_{1}^{i}\). We add a similar shared path gadget \(B_{1}^{i}\) for the \(i^{th}\) row vertex \(b_{1}^{i} \in B_{1}\).
For each edge between \(a_{1}^{i}\) and \(a_{2}^{j}\) in \(G_{x,y}\), we add a corresponding edge
between \(A_{1}^{i}[1]\) and \(a_{2}^{j}\) in \(H_{x,y}\). And similarly for each edge between \(b_{1}^{i}\) and \(b_{2}^{j}\) in \(G_{x,y}\), we
add a corresponding edge between \(B_{1}^{i}[1]\) and \(b_{2}^{j}\) in \(H_{x,y}\). See~\Cref{fig:MVC-lower-bound-dangling-path-gadget} (right) for an illustration.
Note that adding shared path gadgets results in only $O(k)$ additional path gadget vertices.

\begin{figure}
\begin{center}
\includegraphics[width=.9\linewidth]{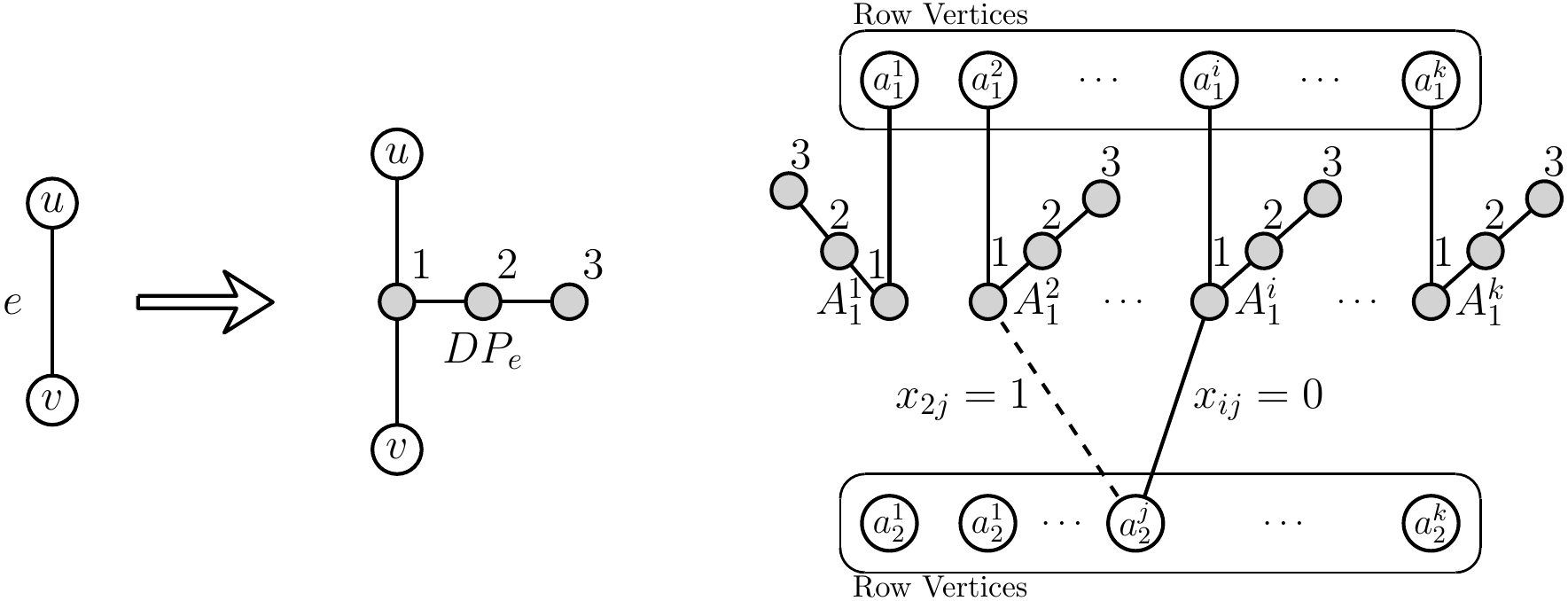}
\end{center}
\caption{\label[figure]{fig:MVC-lower-bound-dangling-path-gadget}This figure shows how an edge \(e\) is replaced by a dangling path gadget \(DP_{e}\) on the left, and on the right it shows how the row vertices in \(A_{1}\) and \(A_{2}\) are connected by shared path gadgets. We show two examples of how we add edges between \(A_{1}^{i}[1]\) and \(a_{2}^{j}\) depending on whether the edge \(\{a_{1}^{i}, a_{2}^{j}\}\) exists in \(G_{x,y}\) or not. The gadget sharing for \(B_{1}\) and \(B_{2}\) is similar. We only show the edges of \(H_{x,y}\) for clarity.}
\end{figure}

\begin{lemma}\label[lemma]{lem:dangling-rearrangement}
Any vertex cover of \(H_{x,y}^2\) of size \(c\) can be modified to a vertex cover of size at most \(c\) that contains no vertex of index \(3\) in any dangling path gadget or shared path gadget. Equivalently, this modified vertex cover contains all vertices in every dangling path gadget except the leaf.
\end{lemma}
\begin{proof}
Fix a particular dangling path gadget or shared path gadget \(P\) such that \(P[3]\) is in the vertex cover. Note that \(P[3]\) is only connected to vertices \(P[1]\) and \(P[2]\) since it is more than \(2\)-hops apart from vertices not in \(P\). Moreover, \(P[1], P[2], P[3]\) form a triangle in \(H_{x,y}^2\). Therefore, any vertex cover must have at least \(2\) vertices from this triangle. So if \(P[3]\) is present in the vertex cover, we can remove it and add any other vertex in the triangle that was excluded (there can be at most one such vertex), and we still cover all the edges in \(H^{2}_{x,y}\).

Since \(P[3]\) is not in the vertex cover, \(P[1], P[2]\) have to be in the vertex cover as these three vertices form a triangle. Doing this process for all dangling path gadgets and shared path gadgets gives us the lemma.
\end{proof}

We now state the reduction from our lower bound graph \(H_{x,y}^2\) to the lower bound graph \(G_{x,y}\) of~\cite{Censor-HillelKP17} in the following lemma.

\begin{lemma}\label{lemma:sizeMVC}
The graph \(H_{x,y}^2\) has a minimum vertex cover of size \(W + 2(2k + 4k\log_2 k + 8\log_2 k)\) if and only if the graph \(G_{x, y}\) has a minimum vertex cover of size \(W\).
\end{lemma}
\begin{proof}
For the forward direction, consider the vertices from the dangling path gadgets in a minimum vertex cover \(C_H\) of \(H_{x,y}^2\). By~\Cref{lem:dangling-rearrangement}, we can assume that for each dangling path gadget and shared path gadget \(P\), the vertices \(P[1],P[2]\) belong to \(C_H\) and the vertex \(P[3]\) does not belong to \(C_{H}\). There are \(2k + 4k\log_2 k + 8\log_2 k\) such gadgets in \(H_{x, y}\). Therefore, \(C_{H}\) contains \(2(2k + 4k\log_2 k + 8\log_2 k)\) vertices which cover all the edges that have a dangling (or shared) path gadget vertex as an endpoint.

The rest of the vertices in \(C_H\) have to cover all the edges formed by pairs of non-gadget vertices that have a dangling (or shared) path gadget connecting them. These are exactly the edges in \(G_{x,y}\) and therefore, the non-gadget vertices in \(C_H\) must form a minimum vertex cover of \(G_{x,y}\), since otherwise we can create a smaller cover of \(H_{x,y}^2\) than \(C_H\) by taking the vertices corresponding to the MVC of \(G_{x,y}\) instead.

For the reverse direction, consider the MVC \(C_G\) of \(G_{x,y}\) having size \(W\). We can take all the \(2(2k + 4k\log_2 k + 8\log_2 k)\) vertices indexed \(1, 2\) from all the dangling path gadgets and shared path gadgets, plus the corresponding \(W\) vertices in \(C_G\) to form a vertex cover of \(H_{x,y}^2\). We cannot form a smaller vertex cover in \(H_{x, y}^2\), because otherwise we could extract a vertex cover of \(G_{x, y}\) that is smaller than \(C_G\) using~\Cref{lem:dangling-rearrangement}, which would contradict the optimality of \(C_G\).
\end{proof}

\begin{proof}[Proof of \Cref{thm:MVC-lower-bound-main-theorem}]
  Censor-Hillel et al.~\cite{Censor-HillelKP17} show that the MVC lower bound graph \(G_{x,y}\) is a family of lower bound graphs for the \CONGEST\ model wrt the set-disjointness function \(f=DISJ_{k^{2}}\) and the predicate \(P_{G}\) which asks whether the graph \(G\) has a vertex cover of size \(W = 4(k - 1) + 4\log k\). The vertices of \(G_{x,y}\) are partitioned into \(V_{A} = A_{1} \cup A_{2} \cup \{f_{S}^{i}, t_{S}^{i} \mid 1 \le i \le \log_{2} k, S \in \{A_{1}, A_{2}\}\}\) and \(V_{B} = V(G_{x,y}) \setminus V_{B}\), with the cut size being \(|E(V_{A}, V_{B})| = O(\log k)\).

  By \Cref{lemma:sizeMVC}, we know that \(G_{x,y}\) satisfies the predicate \(P_{G}\) iff the graph \(H_{x,y}\) satisfies the predicate \(P_{H}\) which asks whether the input graph has a \(G^{2}\)-vertex cover of size \(W + 2(2k + 4k\log_2 k + 8\log_2 k)\). Recall that the number of vertices in \(H_{x,y}\) is \(O(k \log k)\).

  Define \(V_{A}' = V_{A} \cup \{DP_{e} | e = \{u, v\} \text{ and } u, v \in V_{A}\}\cup \{A_{1}^{i} \mid 1 \le i \le k\}\) and \(V_{B}' = V(H_{x,y}) \setminus V_{A}'\). With these definitions of \(V_{A}'\) and \(V_{B}'\), the size of the cut of \(H_{x,y}\) is \(|E(V_{A}', V_{B}')| = O(\log k)\).

  The graph \(H_{x,y}\) with partition \(V_{A}', V_{B}'\) is a family of lower bound graphs wrt the set-disjointness function \(f=DISJ_{k^{2}}\) and the predicate \(P_{H}\). Therefore, \Cref{thm: general lb framework} gives an \(\tilde{\Omega}(k^2)\) lower bound for the problem of deciding whether a graph with \(O(k \log k)\) vertices has a vertex cover of size \(W + 2(2k + 4k\log_2 k + 8\log_2 k)\) where \(W = 4(k - 1) + 4 \log_{2} k\). Therefore for a graph with \(n\)-vertices we get an \(\tilde{\Omega}(n^{2})\) lower bound which completes the proof of \Cref{thm:MVC-lower-bound-main-theorem}.
\end{proof}

\subsection{Limitations of~\Cref{thm: general lb framework}:}\label{sec:VC-apx-limitations}
With the goal of finding how good is the complexity we obtain for a $(1+\eps)$-approximation of MVC in~\Cref{sssec:UBunweightedVC}, we tried to prove a lower bound for it, with respect to $n$ (We mention that the $1/\eps$ term in the complexity is unavoidable, due to a straightforward adaptation of the lower bound given in~\cite{BenBasatKS19}.  The quadratic lower bounds in this paper, as well as in previous papers~\cite{Censor-HillelKP17, BachrachCDELP19}, are all obtained by lower bound graphs of small cuts (logarithmic size). We show that any construction which has a cut of size $o(n)$ cannot give any super-constant lower bound for a $(1+\eps)$-approximation of MVC, for any $\eps = O(1)$. 

\begin{lemma}\label[lemma]{lemma:limitationVC}
Let $P$ be a predicate that implies a $(1+\eps)$-approximation for \(G^{2}\)-MVC. If $\{G_{x,y}\}$ is a family of lower bound graphs with respect to a function $f$ and the predicate $P$, which has a cut $C$ of size $o(n)$, then~\Cref{thm: general lb framework} cannot give a super-constant lower bound for a distributed algorithm for deciding $P$.
\end{lemma}

\begin{proof}
The two players construct a graph from the family according to their inputs $x$ and $y$. Each player takes into the cut all of its cut vertices, denoted $C_A$ and $C_B$, respectively, and whatever other vertices from $V_A \setminus C_A$ (respectively, $V_B\setminus C_A$) that form an optimal cover of the edges of $G^2_{x,y}$ that remain after taking the cut vertices $C_A \cup C_B$. The players inform each other about the number of vertices each one took into the cover and conclude the size of the computed cover. This requires exchanging $O(\log{n})$ bits, thus $CC(f) \leq O(\log{n})$.

Our first claim is that the set of selected vertices is indeed a cover. This is because taking all of the cut vertices $C_A \cup C_B$ into the cover promises that any yet uncovered edge of $G^2_{x,y}$ has both endpoints in $V_A \setminus C_A$ or both in $V_B \setminus C_B$, and thus adding any cover on each side gives a cover for $G^2_{x,y}$.

Second, we claim that the computed cover is a $(1+\eps)$-approximation of an optimal solution. The reason is that the computed cover takes an optimal cover of the edges of $G^2_{x,y}$ that have both endpoints in $V_A \setminus C_A$ or both in $V_B \setminus C_B$, and since these are disjoint sets then an optimal solution must take at least this number of vertices. The computed solution then has to account also for the cut vertices $C_A \cup C_B$. However, these are only $o(n)$, while we know from~\Cref{lem:structAllNodesVC} that the size of an optimal solution is at least $n/2$, which gives an approximation factor of $1+o(1)$.

We are now ready to complete the proof. Suppose that $G_{x,y}$ is used with some function $f$ to show a lower bound for a $(1+\eps)$-approximation for \(G^{2}\)-MVC using~\Cref{thm: general lb framework}. Then this lower bound is $\Omega (CC(f)/|C|\log n)$. But $CC(f) \leq O(\log{n})$ and so no super-constant lower bound can be derived with this approach using small cuts.
 \end{proof}

 \subsection{Conditional Hardness for \((1+\epsilon)\)-Approximation}\label{sec:VChardness}
\begin{theorem}
\label{theorem:HtoG}
Let  $\delta, \alpha$ be constants in $(0,1)$. If for every $\eps \in (0,1)$ there is a $(1+\eps)$-approximation algorithm for MVC of $G^2$ that completes in $O(n^\alpha/\eps)$ rounds, then there is a $(1+\delta)$-approximation algorithm for MVC of $G$ that completes in $\tilde{O}(D + n^{(4(1+\alpha)+2\rho)/3})$ rounds, where $\rho=\log(1/\delta)/\log{n}$.
\end{theorem}

\begin{proof}
Let $ALG$ be a $(1+\eps)$-approximation algorithm for MVC on $G^2$ that completes in $O(n^\alpha/\eps)$ rounds. The high-level goal is to deduce an approximate solution for MVC of $G$ given the approximate solution for MVC of $H^2$, for a related graph $H$. An obstacle in doing so is that we will need to run $ALG$ with a value of $\eps$ that depends on the size of the optimal vertex cover of $G$, which we denote by $OPT$, and for this we need $OPT$ to be sufficiently large. To this end, we will first find a very rough approximation for $OPT$, and if it is not sufficiently large then we resort computing a $(1+\delta)$-approximation for it using the parametrized approach of~\cite{BenBasatKS19}.

Formally, we define $\beta=(2(1+\alpha)+\rho)/3$. We run the 2-approximation algorithm for MVC on $G$ given by~\cite{Ben-BasatEKS18}, which takes $O(\log n\log\Delta/\log^2\log\Delta)$ rounds to complete.\footnote{This is the state-of-the-art for a 2-approximation. We note that we could use here any constant approximation algorithm but we omit $\poly\log n$ factors anyhow.} Within another $O(D)$ rounds the nodes learn the size of the given solution, denoted by $SOL$. Let $\gamma = \log{(SOL/2)}/\log n$, implying that $n^{\gamma} = SOL/2$. 

We now consider two cases, depending on whether $\gamma$ is smaller than $\beta$ or is at least $\beta$. If \(\gamma<\beta\) then we run the \((1+\delta)\)-approximation algorithm for MVC of $G$ given by~\cite{BenBasatKS19}, which takes \(O(n^{2\gamma})\) rounds. Because \(\gamma<\beta\), we have that in this case our algorithm completes within $O(\poly\log(n) + D + n^{2\gamma}) = \tilde{O}(D+ n^{2\beta})$ rounds.

Otherwise, \(\gamma\) is at least \(\beta\). In this case we define a graph \(H\) that is obtained from the graph \(G\) by replacing each edge $e =\{w,u\} $ in $G$ with dangling path gadget $DP_e$, as described in Section~\ref{sec:orgcb69afc}. Recall that $DP_e$ is a path on three vertices \(p_{e}^{1}, p_{e}^{2}, p_{e}^{3}\) which is connected by $p_e^1$ to both $u$ and $w$. We simulate an execution of \(ALG\) for MVC of \(H^2\) with $\eps=\delta n^{\beta}/3m = n^{\beta-\rho}/3m$. By our assumption, executing $ALG$ on $H$ completes in $O(n_H^\alpha/\eps)$ rounds, where $n_H = O(m)$ is the number of nodes in $H$ (here $m$ is the number of edges in $G$). For the simulation of $ALG$ on $H$, the nodes of $G$ simply assign each edge to one arbitrary endpoint (say, the one with the smaller identifier) and each node in $G$ simulates itself and the nodes of the gadgets that correspond to the edges that are assigned to it. Since the simulated nodes of each gadget are only connected with a single edge to the other endpoint of the original edge, this simulation incurs no overhead, thus completes in $O(m^\alpha/\eps)$ rounds.

Let $C$ be the set of the original nodes of \(G\) (the non-gadget nodes) that \(ALG\) takes into the cover $C_H$ of \(H^2\) that it produces. Our first claim is that $C$ is a cover of $G$. This follows since for every edge $e=\{u,w\}$ in $G$, it holds that $\{u,w\}$ is in $H^2$, and thus at least one of its endpoints has to be in $C_H$. Further, we claim that $C$ cannot be too large compared with $OPT$. To see this, note that any cover of $H^2$ must take at least 2 nodes of every gadget, and hence $C \leq C_H - 2m$. Moreover, the size of any optimal solution $OPT_H$ for $H^2$ is exactly $OPT_H = OPT + 2m$, because any smaller solution either does not take 2 nodes of every gadget or induces a cover for $G$ that is smaller than $OPT$, either of which is impossible. 

Hence, we have 
$$C \leq C_H - 2m \leq (1+\eps)OPT_H -2m = (1+\eps)(OPT+2m)-2m = OPT(1+\eps(1+2m/OPT)).$$
This means that the approximation factor we get is 
$$1+ \eps(1+2m/OPT)=1 +(\delta n^{\beta}/3m)(1+2m/n^{\gamma}) \leq 1+(\delta n^{\beta}/m)(m/n^\gamma) = 1+\delta(n^{\beta-\gamma}) \leq 1+\delta.$$

The time the simulation takes is \(O(m^\alpha/\eps)= O(m^\alpha \cdot m/(\delta n^\beta)  ) = O( n^{2+2\alpha-\beta+\rho})\).
Thus, the total number of rounds for the algorithm is $\tilde{O}(D+n^{2\beta} + n^{2+2\alpha-\beta+\rho})$. Since $\beta=(2(1+\alpha)+\rho)/3$, we get a number of rounds which is $\tilde{O}(D+n^{(4(1+\alpha)+2\rho)/3})$.
\end{proof}

In particular, Theorem~\ref{theorem:HtoG} tells us that going below $\alpha=1/2$ for small values of $\eps$ would yield a sub-quadratic algorithm for any constant approximation $(1+\delta)$ for $G$, which would answer a major open question in distributed MVC approximation.
\begin{corollary}
\label{corollary:HtoG}
Let  $\delta$ be a constant in $(0,1)$. If for every $\eps \in (0,1)$ there is a $(1+\eps)$-approximation algorithm for MVC of $G^2$ that completes in $o(n^{1/2}/\eps)$ rounds, then there is a $(1+\delta)$-approximation algorithm for MVC of $G$ that completes in $o(n^2)$ rounds.
\end{corollary}

Corollary~\ref{corollary:HtoG} says that there is still a gap between our $O(n/\eps)$-round algorithm of Section~\ref{sssec:UBunweightedVC} and an algorithm that would imply a non-trivial runtime for $G$. We mention again that the $1/\eps$ term in the complexity is unavoidable, due to a straightforward adaptation of the lower bound given in~\cite{BenBasatKS19}. 

\section{Distributed $G^2$-Minimum Dominating Set}
\label{sec:MDSupper}
\subsection{An $O(\log n)$-Approximation for $G^2$-MDS (Randomized)}\label{sec:MDSUpperBound}

\begin{theorem}
\label{theorem:MDSUpperBound}
  There is a randomized distributed \CONGEST\ model algorithm, which given an input graph \(G\) computes an \(O(\log \Delta)\)-approximate solution to the MDS problem on \(G^{2}\) in \(\text{poly}\log n\) rounds.
\end{theorem}

We simulate the algorithm proposed in~\cite{Censor-HillelD18} for approximating MDS in \(G\). The algorithm guarantees an \(O(\log \Delta)\) approximation in \(O(\log n \log \Delta)\) rounds\footnote{If one is careful with constants in the analysis of~\cite{Censor-HillelD18}, the approximation factor can be shown to be \(8 H_{k}\) where \(H_{k}\) is the \(k^{th}\) harmonic number and \(k \le \Delta^{2}\) is the maximum number of vertices that can be dominated by a single vertex.}. Their algorithm for \(G\) has the following steps in each round:
\begin{enumerate}
\item  Each vertex $v$ computes its rounded density $\rho_v$, where $\rho_{v}$ is the number of uncovered vertices that $v$ covers rounded up to the closest power of $2$. Vertex $v$ sends this value to its 2-hop neighbors in $G$. Here $C_{v}$ is the number of uncovered vertices that $v$ covers.
\item  Each vertex $v$ such that $\rho_v \geq \rho_u$ for each $u$ in its 2-neighborhood is a \textit{candidate}. Vertex $v$ informs its neighbors that it is a candidate.
\item  Each candidate $v$ chooses a random number $r_v \in \{1, \dots, n^4\}$ and sends it to its neighbors.
\item  Each uncovered vertex that is covered by at least one of the candidates, votes for the first candidate that covers it according to the order of the values $r_v$. If there is more than one candidate with the same minimum value, it votes for the one with the minimum ID.
\item  If $v$ receives at least $|C_v|/8$ votes from vertices it covers then it is added to the dominating set.
\item  All vertices that are covered output 0, and $v$ outputs 1 if and only if it was added to the dominating set in the previous step.
\end{enumerate}

We wish to simulate this algorithm on the graph \(G^{2}\) while the network is still \(G\). This poses some interesting congestion problems when \(v\) tries to estimate the number of uncovered \(2\)-hop neighbors of each vertex in \(G\) and the number of votes that it receives from its \(2\)-hop neighbors. The following lemma allows us to get this estimate quickly in a randomized manner.

\begin{lemma}\label[lemma]{lem:sampling-estimation}
  Let \(U \subseteq V\) be an arbitrary set of vertices. If each vertex knows whether or not it belongs to \(U\), it is possible to get an estimate \(\tilde{d_{v}}\) of the quantity \(d_{v} = |N_{2}(v) \cap U|\) for all vertices \(v\) such that with high probabilty
  \[d_{v}(1-\eps) \le \tilde{d_{v}} \le (1+\eps)d_{v}\]
  for a constant \(\eps \in (1,1/4)\), in \(O(\log n)\) rounds in the \CONGEST\ model.
\end{lemma}
\begin{proof}

  We use a simplified version of the estimation algorithm provided in~\cite{Mosk-AoyamaSPODC2006}. In order to estimate \(\sum_{i=1}^{k}{y_{i}}\), the algorithm generates \(r\) independent random \(W_{1}^{i}, \dots, W_{r}^{i}\) such that each \(W_{j}^{i}\) is distributed according to the exponential distribution with mean \(1/y_{i}\) for each \(1 \le i \le k\).

  The algorithm exploits the following property of exponential random variables: for each \(1 \le j \le r\) the random variable \(\tilde{W_{j}} = \min_{1 \le i \le k}{W_{j}^{i}}\) is distributed exponentially with mean \(1/\tilde{y}\) where \(\tilde{y} = \sum_{i=1}^{k}{y_{i}}\). Therefore, the quantity we want to estimate is the reciprocal of the expectation of \(\tilde{W_{j}}\). And the \(r\) independent samples give us concentration around this expectation as shown in the following lemma which is a consequence of Cram\'er's Theorem (\cite{DemboZ2010}, pp. 30, 35).

  \begin{lemma}\label{lem:Mosk-Aoyama-Shah-Estimation}
    Let \(\tilde{W_{1}}, \tilde{W_{2}}, \dots, \tilde{W_{r}}\) be iid exponential random variables with mean \(\lambda\). Let \(\tilde{W} = \frac{1}{r}\sum_{j=1}^{r}{\tilde{W_{j}}}\). Then for any \(\eps \in (0, 1/2)\):
    \(Pr\left( \left| \tilde{W} - \lambda \right| \ge \eps \lambda\right) \le exp(-\eps^{2}r/3)\)
  \end{lemma}

  Therefore, if we use \(r = \log n\), then we get that with high probability, \((1-\eps) 1/\tilde{y} \le \tilde{W} \le (1+\eps)1/\tilde{y}\) which also implies \((1-\eps) \tilde{y} \le 1/\tilde{W} \le (1+2\eps)\tilde{y}\) for \(\eps \in (0,1/2)\).

  Now we have every vertex in \(v \in U\) hold \(y_{v} = 1\). Every vertex \(v \in U\) generates \(r = O(\log n)\) iid exponential random variables \(W_{1}^{v}, \dots, W_{r}^{v}\) with mean \(1\), and broadcasts each random variable to its neighbors in \(r\) rounds of \CONGEST. Once a vertex \(v \in V\) receives \(W_{j}^{u}\) for all neighbors \(u \in N_{1}(v)\), it sends \(\overline{W}_{j}^{v}\min_{u \in N_{1}(v)}{W_{j}^{u}}\) to all its neighbors. Now once a vertex \(v \in V\), receives \(\overline{W}_{j}^{u}\) for all neighbors \(u \in N_{1}(v)\), it calculates \(\tilde{W_{j}} = \min_{u \in N_{1}(v)}{\overline{W}_{j}^{u}} = \min_{u \in N_{2}(v)}{W_{j}^{u}}\). The estimate that each \(v \in V\) outputs is \(\tilde{d_{v}} = \frac{r}{\sum_{j=1}^{r}{\tilde{W_{j}}}}\) where \(1/\tilde{d_{v}}\) is distributed exponentially with mean \(\lambda = 1/d_{v}\). Therefore, the statement of the lemma follows from~\Cref{lem:Mosk-Aoyama-Shah-Estimation}.

  For each vertex \(v \in V\) that computes some estimate \(\tilde{d_{v}}\), the value \(d_{v}\) that \(v\) is estimating lies in the set \(\{1, \dots, \Delta^{2}\}\). We can assume that \(d_{v} \ge c \log n\) for any arbitrarily large constant \(c\), by having vertices with degree \(< c \log n\) broadcast all their edges in \(O(\log n)\) rounds. Therefore, \(O(\log n)\) bits of precision suffice to get the correct estimate with high probability since rounding will only affect the final estimate by an additive \(r = \log n\) factor.
\end{proof}

Using~\Cref{lem:sampling-estimation} with \(U\) being the set of uncovered vertices, each vertex \(v\) can calculate its rounded density \(\tilde{\rho_{v}}\) of step 1 in \(O(\log n)\) rounds. In step 2, each vertex just needs the maximum rounded density in its \(4\)-hop neighborhood in \(G\) in order to mark itself as a candidate. For steps 3 and 4, it suffices that each uncovered vertex \(u\) know the ID of the vertex having minimum rank in their \(2\)-hop neighborhood (where rank ties are broken by smallest ID) in order to know which candidate \(u\) is voting for. For step 5, we wish to estimate the number of votes and \(|C_{v}|\)\footnote{The algorithm in~\cite{Censor-HillelD18} uses the exact value of \(|C_{v}|\). But it suffices to use a good estimate, since it doesn't affect the approximation factor and only changes the running time by a constant.} for each candidate \(v\). Estimating \(|C_v|\) can be done using~\Cref{lem:sampling-estimation} the same way we estimated \(\rho_{v}\). Estimating the number of votes is a bit different. Note that the candidates form a partition of the uncovered vertices, therefore we can apply \Cref{lem:sampling-estimation} for each candidate in parallel which allows the candidates to estimate the number of votes that they have received. Note that when performing this estimation, a vertex might receive vote estimates for many different candidates that it needs to forward and it will send the estimate only to the candidate it corresponds to instead of broadcasting it like in the proof of~\Cref{lem:sampling-estimation}. This allows us to simulate step 5 in \(O(\log n)\) rounds. For step 6, it suffices that each uncovered node \(u\) know the smallest ID vertex in its \(2\)-hop neighborhood that joins the dominating set.

\section{Distributed $G^2$-Minimum Dominating Set (Lower Bounds)}
\label{sec:MDSLB}
In this section we show \(\tilde{\Omega}(n^2)\) lower bounds in the \CONGEST~model for solving the (unweighted) $G^2$-MDS problem. In \Cref{app:ExactMDSLB} we show a lower bound for solving the problem exactly and  in \Cref{app:ApproximateMDSLB} and \Cref{app:ApproximateMDSLB-unweighted} we show lower bounds for computing constant approximations.

\subsection{Quadratic Lower bound for Exact $G^2$-MDS}\label{app:ExactMDSLB}
In this section we will show an \(\tilde{\Omega}(n^2)\) lower bound for solving the (unweighted) $G^2$-MDS problem exactly in  \CONGEST, that is, we formally prove the following theorem. 
\begin{theorem}
  \label[theorem]{thm:exact-mds-lower-bound-main-theorem}
	Any \CONGEST\ algorithm requires \(\tilde{\Omega}(n^{2})\) rounds for solving (unweighted) $G^2$-MDS exactly.
\end{theorem}

We now give an outline for the proof of \Cref{thm:exact-mds-lower-bound-main-theorem}. The formal proof follows at the end of the section.

\smallskip

\textbf{Proof Outline of \Cref{thm:exact-mds-lower-bound-main-theorem}: }
Bachrach et al.~\cite{BachrachCDELP19} provided a family of lower bound graphs \(G_{x,y}\) that shows that solving exact $G$-MDS needs near quadratic time in the \CONGEST\ model. To prove \Cref{thm:exact-mds-lower-bound-main-theorem} we construct a graph family $H_{x,y}$ such that the size of the MDS in \(G_{x,y}\) is closely related to the size of an exact MDS in \(H^2_{x,y}\) (cf. \Cref{lem:MDSsizesRelate}). Then, one can solve $G$-MDS on $G_{x,y}$ via creating the graph $H_{x,y}$ and then solving MDS on $H^2_{x,y}$. Thus a lower bound for $G$-MDS translates into a lower bound for the $G^2$-MDS problem. The main difficulty is in having asymptotically the same number of vertices in $H$ graphs while keeping the graph $H$ simulatable in the communication network $G$. If $H$ has drastically more vertices that $G$ the lower bounds results for $G^2$-MDS would be very far from being quadratic.

We start with the lower bound graph family used by~\cite{BachrachCDELP19} to show that
$G$-MDS requires $\tilde{\Omega}(n^2)$ rounds to be solved exactly.

\begin{figure}
\begin{center}
\includegraphics[width=.9\linewidth]{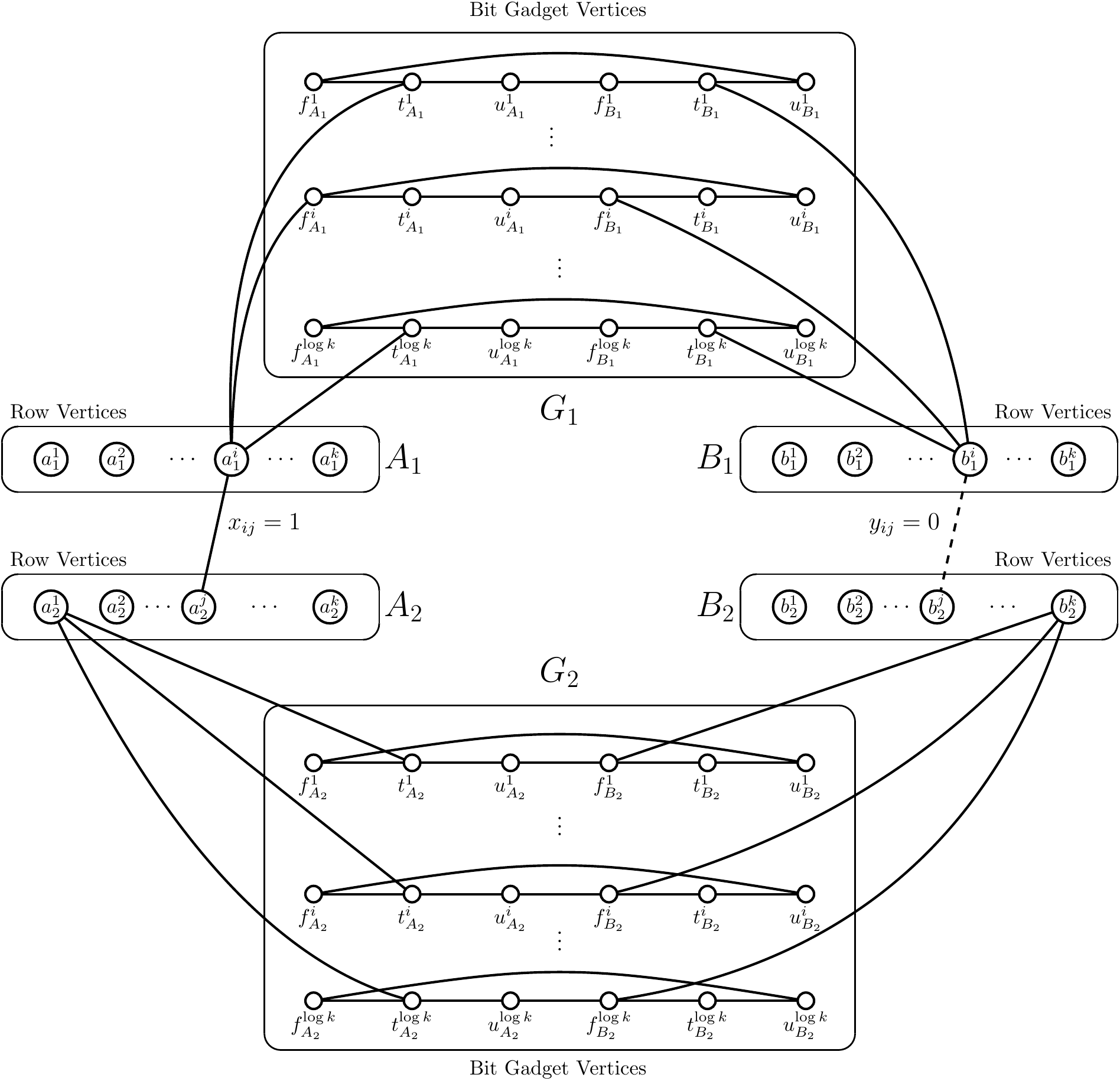}
\end{center}
\caption{\label[figure]{fig:Bachrach-MDS-lb-graph}This figure shows the lower bound graph \(G_{x,y}\) that appears in \cite{BachrachCDELP19} used to show a quadratic lower bound for computing exact MDS in the \CONGEST\ model. We use this graph as the basis for our MDS lower bounds}
\end{figure}

\textbf{$G$-MDS lower bound graph family \(G_{x,y}\) from~\cite{BachrachCDELP19}: }
For any $k$ that is a power of 2 and for each pair of bit vectors $x, y \in \{0, 1\}^{k^2}$, there is a graph, denoted $G_{x, y}$, in this family.
See \Cref{fig:Bachrach-MDS-lb-graph} for an illustration of $G_{x,y}$.
This lower bound graph has four sets of row vertices \(A_{1}, B_{1}, A_{2}, B_{2}\) each of which contain \(k\) vertices.
Moreover there are two sets of \textit{bit gadgets}, each set containing \(\log_{2} k\) bit gadgets, one set for \(A_{1}, B_{1}\) and the other set for \(A_{2}, B_{2}\).
The \(i^{th}\) bit gadget for \(A_{1}, B_{1}\) is a \(6\)-cycle with vertices \(f_{A_{1}}^{i}, t_{A_{1}}^{i}, u_{A_{1}}^{i}, f_{B_{1}}^{i}, t_{B_{1}}^{i}, u_{B_{1}}^{i}\).
The vertices in \(a_{1}^{i} \in A_{1}\) are connected to the bit gadget vertices \(f_{A_{1}}^j, t_{A_{1}}^j\) depending on the binary representation of \(i-1\).
Specifically, $a_1^i$ is connected to the complement of the binary representation of $i-1$.
For example, the vertex \(a_{1}^{1}\) is connected to all the \(t_{A_{1}}\) vertices. The connections for other row vertices are similar.
All of these edges are fixed, i.e., independent of $x$ and $y$. Additionally, $x$ determines edges between $A_1$ and $A_2$, whereas $y$ determines edges between
$B_1$ and $B_2$. Specifically, index the $k^2$ bits in $x$ as $x_{i, j}$, $1 \le i, j \le k$. Connect vertex $i \in A_1$ and vertex $j \in A_2$ iff $x_{i,j} = 1$.
The edges between $B_1$ and $B_2$ are similarly determined by the bit vector $y$.
One can check that $G_{x, y}$ has $4k + 12\log_2 k$ vertices, $4k\cdot \log_2 k$ fixed edges, and $O(k^2)$ variable edges (i.e., edges determined
by $x$ and $y$).

\cite{BachrachCDELP19} now defines a vertex partition $(V_A, V_B)$ of $G_{x, y}$, where $V_A = A_1 \cup A_2 \cup \{t_{A_1}^i, f_{A_1}^i, u_{A_1}^i \mid 1 \le i \le \log_2 k\} \cup \{t_{A_2}^i, f_{A_2}^i, u_{A_2}^i \mid 1 \le i \le \log_2 k\}$ and $V_B$ is the set of remaining vertices.
Basically, the vertices in left half of \Cref{fig:Bachrach-MDS-lb-graph} are assigned to $V_A$ and those in the right half are assigned to $V_B$.
Now consider two players Alice and Bob and suppose $V_A$ (and incident edges) are provided to Alice and $V_B$ and incident edges are provided
to Bob. Bachrach et al.~\cite{BachrachCDELP19} show that the construction of $G_{x, y}$ is such that $G_{x, y}$ has a dominating set of size at most $4\log_2 k + 2$ iff
 $DISJ_{k^{2}}(x,y) = \mathrm{false}$ for the bit vectors $x$ and $y$. Since the two-party communication complexity of set disjointness for bit vectors of size $k^2$ is
$\Omega(k^2)$, Alice and Bob need to communicate $\Omega(k^2)$ bits to determine if $G_{x, y}$ has a dominating set of size at most $4\log_2 k + 2$.
The number of edges in the cut between Alice and Bob is $O(\log k)$, implying that if our goal was to determine if $G_{x, y}$ has a
dominating set of size at most $4\log_2 k + 2$ in the \CONGEST\ model, then $\Omega(k^2)$ bits would have to flow over $O(\log k)$ edges,
leading to a $\tilde{\Omega}(k^2)$ lower bound on the number of rounds.

We introduce a \emph{dangling path gadget} which we insert into edges in the graph $G_{x,y}$ to obtain a
graph $H_{x,y}$. The goal is to show that one can solve MDS on $G_{x,y}$ by solving $G^2$-MDS on $H_{x,y}$.
Note that the motivation for introducing a dangling path gadget into edge $e$ of $G_{x, y}$ is to ensure that
$H^2_{x, y}$ has all the edges of $G_{x, y}$ and we can compute a minimum dominating set of $G_{x, y}$ by computing
a minimum dominating set of $H^{2}_{x, y}$ and exchanging/removing vertices that
cover the gadgets.
Here is a more precise description of the dangling path gadget.

\textbf{The dangling path gadget \(DP_{e}\):} We propose adding the following dangling path gadget \(DP_{e}\) replacing each edge \(e = (u, v)\) in \(G_{x,y}\).
The gadget has \(5\) vertices which we denote by \(DP_e[1]\), \(DP_e[2]\), \(DP_e[3]\), \(DP_e[4]\), \(DP_e[5]\).
Then edge $e$ is deleted, the vertex \(DP_e[1]\) is connected to \(u\) and \(v\), and
there is a path connecting \(DP_e[1], DP_e[2], DP_e[3], DP_e[4], DP_e[5]\).
See Figure~\ref{fig:Exact-MDS-dangling-path-gadgets} (left) for an illustration.
\begin{figure}
\begin{center}
\includegraphics[width=.83\linewidth]{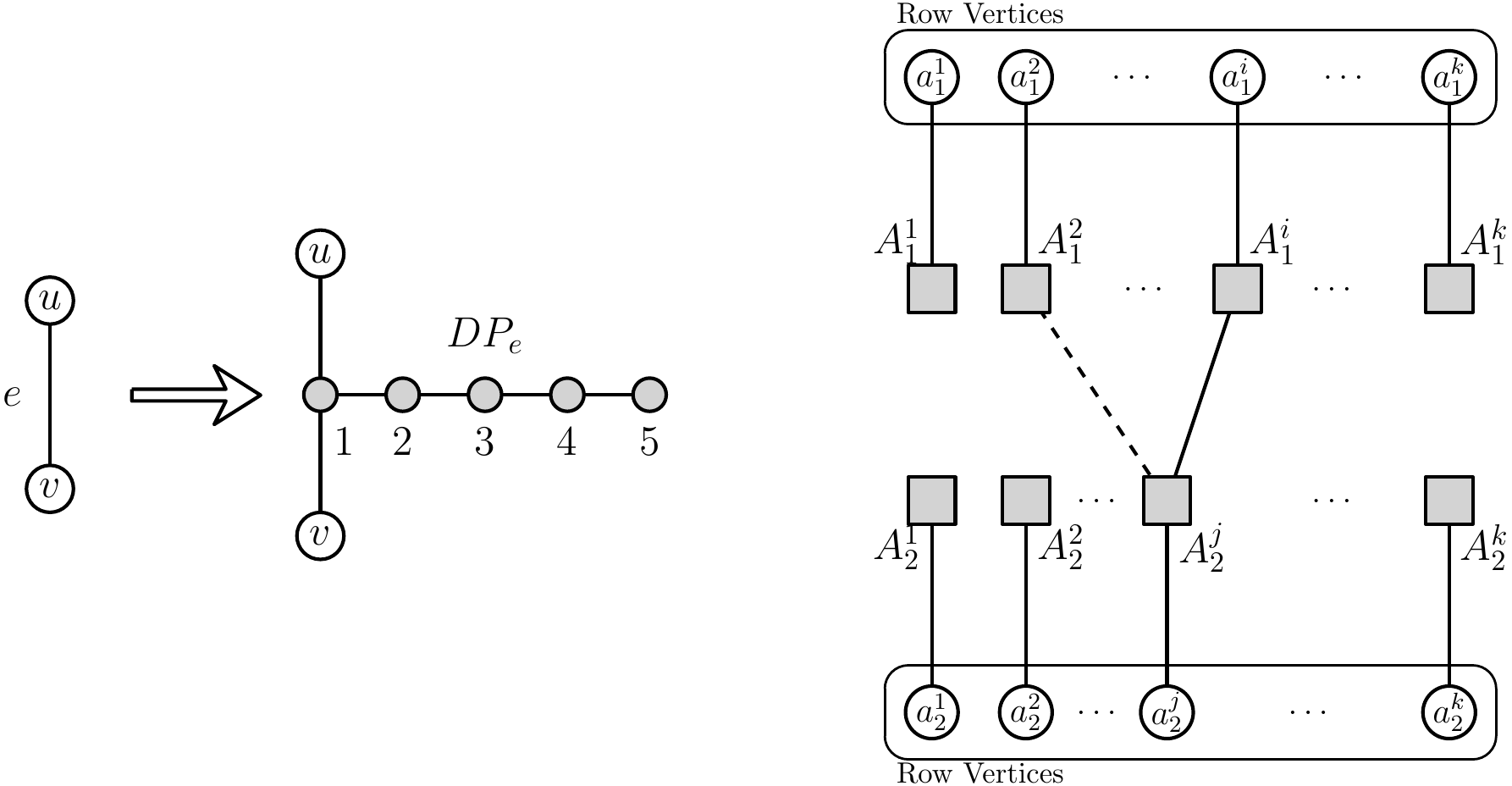}
\end{center}
\caption{\label[figure]{fig:Exact-MDS-dangling-path-gadgets}This figure shows how an edge \(e\) is replaced by a dangling path gadget \(DP_{e}\) on the left, and on the right it shows how the row vertices in \(A_{1}\) and \(A_{2}\) have shared path gadgets attached. We show two examples of how we add edges between \(A_{1}^{i}[1]\) and \(A_{2}^{j}[1]\) depending on whether the edge \(\{a_{1}^{i}, a_{2}^{j}\}\) exists in \(G_{x,y}\) or not. We only show some edges of \(H_{x,y}\) for clarity.}
\end{figure}

There is one main challenge posed by this approach, which we now describe along with a description of how we get around it.
\begin{quote}
\textbf{Challenge 1:}
Introducing a dangling path gadget into each edge of the graph increases the number of vertices of the graph quite significantly.
\end{quote}

\noindent
If $G_{x, y}$ is a graph with $n$ vertices and $m$ edges, then $H_{x, y}$ has $\Omega(m)$ vertices which could be $\Omega(n^2)$ vertices in the worst case.
This ``blow up'' in the number of vertices means
the $\tilde{\Omega}(n^2)$ lower bound for $G$-MDS would only translate into a $\tilde{\Omega}(n)$ lower bound for $G^2$-MDS.
To deal with this challenge, we introduce the idea of \textit{sharing} dangling path gadgets. Informally speaking, this simply means that instead of each edge having a separate dangling path gadget, a lot of the edges in $G_{x, y}$ will share dangling path gadgets. A precise version of the construction of a lower bound graph family, that uses this idea, is described below.

\textbf{$G^{2}$-MDS lower bound graph family $H_{x,y}$:} We first replace each edge having at least one bit gadget vertex as
an end point, by a \(5\)-vertex dangling path gadget.
Note that there are $O(k \log k)$ edges incident on bit gadget vertices and therefore we have not introduced too many path gadget vertices.
The remaining edges are those between the row vertices \(A_{1}, B_{1}, A_{2}, B_{2}\).
These can be $O(k^2)$ in number and so we have to be careful in introducing path gadget vertices.
For each row vertex \(a_{1}^{i} \in A_{1}\), we add a \textit{shared path gadget} \(A_{1}^{i}\).
The gadget is similar to the dangling path gadget in that it has \(5\) vertices
\(A_{1}^{i}[1], A_{1}^{i}[2], A_{1}^{i}[3], A_{1}^{i}[4], A_{1}^{i}[5]\) connected to form a path.
The vertex \(A_{1}^{i}[1]\) is connected to \(a_{1}^{i}\). We add similar shared path gadgets
\(A_{2}^{i}, B_{1}^{i}, B_{2}^{i}\) for the \(i^{th}\) row vertex in \(A_{2}, B_{1}\), and \(B_{2}\) respectively.
For each edge between \(a_{1}^{i}\) and \(a_{2}^{j}\) in \(G_{x,y}\), we add a corresponding edge
between \(A_{1}^{i}[1]\) and \(A_{2}^{j}[1]\) in \(H_{x,y}\). Similarly for each edge between \(b_{1}^{i}\) and \(b_{2}^{j}\) in \(G_{x,y}\), we
add a corresponding edge between \(B_{1}^{i}[1]\) and \(B_{2}^{j}[1]\) in \(H_{x,y}\). See~\Cref{fig:Exact-MDS-dangling-path-gadgets} for an illustration.
Note that this sharing of path gadgets results in only $O(k)$ additional path gadget vertices.

We now show, in a sequence of three lemmas that any MDS on $H^2_{x,y}$ can be put into a normal form (Lemmas \ref{lem:middle-dangling-gadget-in-mds} and \ref{lem:no-other-dangling-gadget-in-mds}). Afterwards we use this normal form to show how the size of an MDS of $G_{x,y}$ is related to the size of an MDS on $H^2_{x,y}$ (\Cref{lem:MDSsizesRelate}).

\begin{lemma}
\label[lemma]{lem:middle-dangling-gadget-in-mds}
Any MDS of $H^2_{x,y}$ can be transformed into an equal size MDS such that
\begin{enumerate}
\item the vertex \(DP_{e}[3]\) of each dangling path gadget \(DP_{e}\) is in the MDS of \(H^2_{x,y}\),
\item \(S[3]\) of the shared path gadgets \(S \in \{A_{1}^{i}, A_{2}^{i}, B_{1}^{i}, B_{2}^{i}\}_{i=1}^{k}\) belongs to the MDS of \(H^{2}_{x,y}\).
\end{enumerate}
\end{lemma}
\begin{proof}
  Consider a dangling path gedget or a shared path gadget \(P\) in \(H^{2}_{x,y}\). In order to cover \(P[5]\), at least one of \(P[3], P[4], P[5]\) has to be in the dominating set \(S_{H}\) of \(H^{2}_{x,y}\). If either of \(P[4]\) or \(P[5]\) (or both) is in the dominating set \(S_{H}\), we can create a new dominating set \(S'_{H}\) by removing them and adding \(P[3]\) (if it is not already present). The dominating set \(S'_{H}\) still covers all the vertices in \(H^2_{x,y}\) and has size at most the size of \(S_{H}\). Doing this exercise for all dangling path gadgets and shared path gadgets gives us the lemma.
\end{proof}

\begin{lemma}
\label[lemma]{lem:no-other-dangling-gadget-in-mds}
Any MDS of $H^2_{x,y}$ can be transformed into an equal size MDS such that
\begin{enumerate}
\item 
for any dangling path gadget \(DP_{e}\) such that \(e\) is incident on a bit gadget vertex in \(G_{x,y}\) no gadget vertices other than \(DP_{e}[3]\) belongs to the MDS
\item the vertices \(S[2], S[4], S[5]\) of a shared path gadget \(S\) do not belong to the MDS of \(H^{2}_{x,y}\) for all \(S \in \{A_{1}^{i}, A_{2}^{i}, B_{1}^{i}, B_{2}^{i}\}_{i=1}^{k}\).
\end{enumerate}
\end{lemma}
\begin{proof}
\begin{enumerate}
\item 
Consider a dangling path gadget \(DP_{e}\) for the edge $e = \{u, v\}$ in $G_{x,y}$. 
For ease of exposition, we rename the vertices of \(DP_{e}\) to be \(p, q, r, s, t\), with $p$ connected to $u$ and $v$ in $H_{x, y}$. 
By \Cref{lem:middle-dangling-gadget-in-mds}, \(r\) belongs to the MDS and covers all vertices \(p, q, r, s, t\) in the gadget. 
Note that if \(s\) and \(t\) are in the MDS, they can just be removed, because they do not cover any more vertices than \(r\) in \(H^2_{x,y}\). 
The vertex \(q\) covers just the vertices \(u, v, p, q, r, s\) in \(H^2_{x,y}\) and thus the only vertices \(q\) covers other than \(r\) are \(u\) and \(v\). 
Therefore, if \(q\) is in a minimum dominating set \(S\), we can exchange it for either \(u\) or \(v\) and we still have a dominating set of size at most \(|S|\). 
Similarly, the additional vertices that \(p\) covers over \(r\), are \(u, v\), and all the \(P[1]\) vertices of each dangling path gadget and shared path gadget \(P\) 
incident on \(u\), and \(v\). The all these vertices are covered due to \Cref{lem:middle-dangling-gadget-in-mds}. Therefore, if \(p\) is in the minimum dominating set \(S\), we can exchange it for either \(u\) or \(v\) and we still have a dominating set of size at most \(|S|\).

Hence, we can assume no vertex in \(DP_{e}\) other than \(r\) belongs to the MDS. Repeating this argument for all dangling path vertices gives us the lemma.
\item Note that the last part of the above proof does not apply to the shared path gadgets, though the rest does. 
For example, the \(A_{1}^{i}[1]\) vertex might cover the vertex \(a_{2}^{j}\) if \(\{A_{1}^i[1], A_{2}^{j}[1]\}\) is an edge in $H_{x, y}$. 
Therefore we get the following, slightly weaker, lemma for shared path gadgets. 
\end{enumerate}
\end{proof}

We are now ready to prove the lemma that will allow us to show our lower bound.

\begin{lemma}
\label{lem:MDSsizesRelate}
The graph \(H_{x,y}^2\) has a minimum dominating set of size \(W + 2k + 4k\log_2 k + 12\log_2 k\) if and only if, the graph \(G_{x, y}\) has a minimum dominating set of size \(W\). 
\end{lemma}
\begin{proof}
For the forward direction, let $S_H$ denote a minimum dominating set of $H_{x, y}^2$.
By \Cref{lem:middle-dangling-gadget-in-mds}, we 
can assume that the vertex with index 3 from each dangling path gadget and each shared path gadget belongs to \(S_H\). 
There are \(2k + 4k\log_2 k + 12\log_2 k\) such vertices in \(H_{x, y}\) and they only cover the dangling path gadget and 
shared path gadget vertices. 
Let $S_H'$ denote the remaining vertices of \(S_{H}\), i.e., those that remain in $S_H$ after vertices with index 3 from
dangling path and shared path gadgets are removed.
Let $W$ denote the size of $S_H'$.
We show that $S_H'$ can be transformed into a minimum dominating set of \(G_{x,y}\) of size $W$.

By \Cref{lem:no-other-dangling-gadget-in-mds}, we know that \(S_{H}\) does not contain any dangling path gadget vertex 
besides those with index 3. Therefore $S_H'$ contains no dangling path vertex.
By \Cref{lem:no-other-dangling-gadget-in-mds}, (2), we know that \(S_{H}\) does not contain the \(S[2], S[4], S[5]\) vertices of any shared path gadget \(S\).
Therefore, for any shared path gadget $S$, the only vertices from $S$ that $S_H'$ may contain is $S[1]$.
Now these vertices in \(S_H'\) have to cover all the vertices in \(H^2_{x,y}\) that are not dangling and shared path vertices;
these are exactly the vertices of \(G_{x,y}\). 

We now show that $S_H$ does not contain both $a_1^i$ and $A_1^i[1]$ for any $i$, $1 \le i \le k$.
The same argument can be applied to other sets of row vertices $A_2$, $B_1$, and $B_2$.
The vertex \(a_{1}^{i}\) in \(H^{2}_{x,y}\) covers some bit gadget vertices and the vertex \(A_{1}^{i}[1]\) covers the 
same row vertices in \(A_{2}\) as the vertex \(a_{1}^{i}\) does in \(G_{x,y}\). Note that \(a_{1}^{i}\) does not cover any row vertex in \(H^{2}_{x,y}\).
We know from~\cite{BachrachCDELP19} that the MDS of \(G_{x,y}\) has the property that the bit gadget vertices provide 
coverage for all bit gadget vertices. This is a local argument and it also holds for the MDS of \(H^{2}_{x,y}\) 
since the subgraph induced by the bit gadget vertices in \(H^{2}_{x,y}\) is the same as the subgraph induced in \(G_{x,y}\). 
Therefore, \(a_{1}^{i}\) and $A_1^i[1]$ cannot belong to \(S_{H}\) because if they do, we can remove 
$a_1^i$ and still have a dominating set of $H_{x, y}$, contradicting the fact that $S_H$ is a minimum dominating set of $H^2_{x, y}$.

Knowing that $S_H'$ does not contain both $a_1^i$ and $A_1^i[1]$ (and similarly for vertices from other rows), we can
transform \(S_{H}'\) to the set \(S_{G}\) by replacing any shared path vertex by the corresponding row vertex.
This gives a dominating set of $G_{x, y}$ of size $|S_H'|$.
Note that \(S_G\) must form a \textit{minimum} dominating set of \(G_{x,y}\). 
Otherwise we can create a smaller dominating set of \(H_{x,y}^2\) than \(S_H\), by taking the vertices corresponding to the 
MDS of \(G_{x,y}\) and applying the reverse transformation that replaces all the row vertices by their corresponding shared path gadget vertex.

For the reverse direction, consider an MDS \(S_G\) of \(G_{x,y}\) of size \(W\). 
Let $S_H$ contain the \(P[3]\) vertices from all the \(2k + 4k\log_2 k + 12\log_2 k\) dangling path and shared path gadgets \(P\).
We now add to \(S_{H}\), the \(W\) vertices in \(S_G\), while replacing each the row vertex by its neighboring shared path gadget vertex. 
It is easy to see that \(S_{H}\) has size \(W + 2k + 4k\log_2 k + 12\log_2 k\) and is
a dominating set of \(H_{x,y}^2\). There cannot be a smaller dominating set of  \(H_{x, y}^2\), because otherwise we could extract 
a dominating set of \(G_{x,y}\) that is smaller than \(S_G\) using the procedure described while proving the forward direction, contradicting the optimality of \(S_G\).
\end{proof}
Using that the size of a $G^{2}$-MDS of $H_{x,y}$ relates the size of an MDS of $G_{x,y}$ we can prove \Cref{thm:exact-mds-lower-bound-main-theorem} as follows.
\begin{proof}[Proof of \Cref{thm:exact-mds-lower-bound-main-theorem}]
  Bachrach et al. \cite{BachrachCDELP19} show that the MDS lower bound graph \(G_{x,y}\) is a family of lower bound graphs for the \CONGEST\ model wrt the set-disjointness function \(f=DISJ_{k^{2}}\) and the predicate \(P_{G}\) which asks whether the graph \(G\) has a dominating set of size \(W = 4\log k + 2\). The vertices of \(G_{x,y}\) are partitioned into \(V_{A} = A_{1} \cup A_{2} \cup \{f_{S}^{i}, t_{S}^{i}, u_{S}^{i} \mid 1 \le i \le \log_{2} k, S \in \{A_{1}, A_{2}\}\}\) and \(V_{B} = V(G_{x,y}) \setminus V_{A}\) with cut size being \(|E(V_{A}, V_{B})| = O(\log k)\).

  By \Cref{lem:MDSsizesRelate}, we know that \(G_{x,y}\) satisfies the predicate \(P_{G}\) iff the graph \(H_{x,y}\) satisfies the predicate \(P_{H}\) which asks whether the input graph has a \(G^{2}\)-dominating set of size \(W + 2k + 4k\log_2 k + 12\log_2 k\) where \(W = 4\log k + 2\). Recall that the number of vertices in \(H_{x,y}\) is \(O(k \log k)\).

  Define \(V_{A}' = V_{A} \cup \{DP_{e} \mid e = \{u, v\} \text{ and } u, v \in V_{A}\} \cup \{A_{1}^{i}, A_{2}^{i} \mid 1 \le i \le k\}\) and \(V_{B}' = V(H_{x,y}) \setminus V_{A}'\). With these definitions of \(V_{A}'\) and \(V_{B}'\), the cut \(E(V_{A}', V_{B}')\) of \(H_{x,y}\) has size at most \(O(\log k)\).

  The graph \(H_{x,y}\) with partition \(V_{A}', V_{B}'\) is a family of lower bound graphs wrt the set-disjointness function \(f=DISJ_{k^{2}}\) and the predicate \(P_{H}\). Therefore, \Cref{thm: general lb framework} gives an \(\tilde{\Omega}(k^2)\) lower bound for the exact unweighted \(G^{2}\)-MDS problem in the \CONGEST\ model on a graph with \(O(k \log k)\) vertices. Therefore for a graph with \(n\)-vertices we get an \(\tilde{\Omega}(n^{2})\) lower bound which completes the proof of \Cref{thm:exact-mds-lower-bound-main-theorem}.
\end{proof}
 \subsection{Quadratic Lower bound for $O(1)$-approximate $G^2$-WMDS}\label{app:ApproximateMDSLB}
In this section, we will prove the following theorem.

\begin{theorem}\label[theorem]{thm:approx-MWDS-lower-bound-main-theorem}
  Any distributed algorithm in the \CONGEST\ model which, given an input graph \(G\), produces a \(c\)-approximate solution to the minimum weighted dominating set problem on \(G^{2}\) for \(c < 7/6\) requires \(\tilde{\Omega}(n^{2})\) rounds.
\end{theorem}

Obtaining quadratic lower bounds for approximation algorithms seems much more challenging than obtaining quadratic lower
bounds for algorithms that solve problems exactly.
This is illustrated in~\cite{BachrachCDELP19}, which contains quadratic lower bounds for exact versions of a number
of problems (e.g., MDS, Hamiltonian path, Steiner tree, and max-cut) and a quadratic lower bound for a $O(1)$-approximation algorithm for just
one problem: maximum independent set (\textsc{MaxIS}).
The authors use an interesting gadget, called a \textit{code gadget}, that helps in creating the ``gap'' needed
for the \textsc{MaxIS} problem.
But, the success of this gadget seems to depend a lot on the structure of the \textsc{MaxIS} problem.
In particular, this gadget does not seem to work for MDS and~\cite{BachrachCDELP19} does not show
any lower bounds for approximating MDS on $G$.
However,~\cite{BachrachCDELP19} does show weaker lower bounds (i.e., linear or worse)
for $G^2$-MWDS, though for larger approximation factors.
Specifically, they show two results for $G^2$-MWDS in \CONGEST: (i) an $\Omega(n^{1-\eps}/\log n)$-round lower bound for $O(\eps \log n)$-approximation and (ii) an $\tilde{\Omega}(n)$-round lower bound for $O(\log\log n)$-approximation.
Our results overcome the weaknesses of the results of~\cite{BachrachCDELP19}: Our bounds are indeed quadratic lower bounds for approximation of $G^2$-MDS, and  our results remove the necessity of weights.

We now provide a proof of Theorem \ref{thm:approx-MWDS-lower-bound-main-theorem}, while
providing intuition for the main challenges our proof overcomes in addition to the challenge addressed in \Cref{app:ExactMDSLB}.

\begin{quote}
\textbf{Challenge 2:}
The size of the MDS in \(H^2_{x,y}\) is too large for any small-cut bit gadget to yield a constant fraction gap.
\end{quote}

\noindent
At this stage, there are still $\Omega(k \log k)$ distinct dangling and shared path gadgets and this leads to
a minimum dominating set of size $\Omega(k \log k)$ because every path gadget needs at least 1 distinct vertex
in the dominating set.
To get a lower bound for $O(1)$-approximation, one would need to create a ``gap'' of size $\Omega(k \log k)$
and this is not possible, given the $O(\log k)$ size of the cut.
To overcome this challenge, we propose an ``extreme'' version of sharing path gadgets, which we now describe.

To overcome this challenge, we need to reduce the size of the MDS significantly.
For this purpose, we propose a merged version of the dangling and shared path gadgets which will
use fewer vertices to cover all the gadget vertices. Let \(C\) be an arbitrary set of dangling and shared path
gadgets \(P_{e}\) that were added during the construction of $H_{x, y}$. In order to merge these gadgets, we remove all the
\(P_{e}[3], P_{e}[4], P_{e}[5]\) vertices and connect all the \(P_{e}[2]\) vertices to a common \(3\)
vertex path \(P_{C}[3], P_{C}[4], P_{C}[5]\). These three common vertices play the same role as the
original \(P_{e}[3], P_{e}[4], P_{e}[5]\) vertices for each of the constituent gadgets we merged.
We denote this merged path gadget as \(P_{C}\). Therefore we get the following lemma.
The proof is similar to those of \Cref{lem:middle-dangling-gadget-in-mds} so we skip it.

\begin{lemma}\label[lemma]{lem:merged-gadget-vertex-in-mds}
Let \(\mathcal{C}\) be an arbitrary partition of the dangling and shared path gadgets in $H_{x, y}$. We modify $H_{x, y}$ by merging the dangling and shared path gadgets for each \(C \in \mathcal{C}\). Any MDS of \(H^2_{x,y}\) can be transformed into an equal size MDS of \(H^2_{x,y}\)  which contains the \(P_{C}[3]\) vertex of each merged path gadget \(P_{C}\).
\end{lemma}

This ``extreme'' merging of path gadgets allows us to reduce the size of a MDS substantially\footnote{We note that this ``extreme'' merging fails for \(G^{2}\)-MVC; these different path gadgets,
even after being merged, require a large vertex cover. This may indicate why we do not have
an $\tilde{\Omega}(n^2)$ lower bound for approximating \(G^{2}\)-MVC and also indicates something fundamentally different about
the two problems.}, to $O(\log k)$
But, it is still not clear how to create a large enough ``gap'' i.e., how to ensure that the
size of MDS changes by $\Theta(\log k)$ as a result of small changes in the edges of
$H_{x, y}$ caused by changes to the bit vectors $x$, $y$.
Our goal now is to modify the construction of the lower bound graph so as to
reduce the size of the MDS even further. But, for that we have the overcome the following challenge.

\begin{quote}
\textbf{Challenge 3:}
The bit gadgets themselves contribute $\Omega(\log k)$ vertices to any minimum dominating set.
\end{quote}

In order to address this challenge, we replace the bit gadgets from the exact MDS construction with a \textit{set gadget} \(G_{MDS}\) inspired by the lower bound graph for \(2\)-MDS from~\cite{BachrachCDELP19}. See~\Cref{fig:MDS-weighted-bit-gadget} for an illustration.

\textbf{Set Gadgets:}
Consider a set system in a universe \(\mathcal{U} = \{1, 2, \dots, \ell\}\) of \(\ell\) elements. In \(G_{MDS}\) there are \(T\) vertices corresponding to sets \(S_{1}, S_{2}, \dots, S_{T} \subset \mathcal{U}\) and another \(T\) vertices corresponding to their complements \(\overline{S_{1}}, \overline{S_{2}}, \dots, \overline{S_{T}}\). A vertex has the same name as the set it represents. There are \(2\ell\) vertices \(\{\alpha_{i}, \beta_{i}\}_{i=1}^{\ell}\) where each \(\alpha_{i}, \beta_{i}\) corresponds to the element \(i\) in the universe \(\mathcal{U}\). There are edges between \(\alpha_{i}\) and \(\beta_{i}\) for each \(i \in \mathcal{U}\). And there are membership edges between \(S_{j}\) and \(\alpha_{i}\) if \(i \in S_{j}\) and between \(\overline{S_{j}}\) and \(\beta_{i}\) if \(i \notin S_{j}\). We require a collection of sets \(S_{1}, \dots, S_{T}\) that satisfy the following property which is used to prove hardness of approximation for set cover in different models \cite{LundY1994, Nisan2002}.

\begin{definition}
[\(r\)-covering property] Consider a collection of \(r\) sets \(\mathcal{R}\) from \(\{S_{i}, \overline{S_{i}}\}_{i=1}^{T}\), such that for each index \(i\) the set \(S_{i}\) and its complement \(\overline{S_{i}}\) are not both included together in \(\mathcal{R}\). The sets \(S_{1}, \dots, S_{T}\) are said to satisfy the \(r\)-covering property if for any such \(\mathcal{R}\), there is at least one element in the universe \(\mathcal{U}\) that is not covered by \(\mathcal{R}\) (i.e.~the element does not belong to any set in \(\mathcal{R}\)).
\end{definition}

\begin{lemma}[\cite{Nisan2002}]
  For any \(r \le \log \ell - O(\log \log \ell)\) there exists sets \(S_{1}, \dots, S_{T}\) satisfying the \(r\)-covering property with \(T = e^{\ell/r 2^{r}}\).
\end{lemma}

Therefore, if we consider \(r\) to be some large constant, we have \(\ell = O(\log T)\). We set the weight of all the \(\alpha_{i}\)'s and \(\beta_{i}\)'s in \(G_{MDS}\) to be \(r\), and all other vertices have weight \(1\). Finally we add two vertices \(\alpha\) and \(\beta\) both having weight \(r\). The vertex \(\alpha\) is connected to all \(S_{i}\)'s and \(\beta\) to all \(\overline{S_{i}}\)'s. This also means that all the \(S_{i}\)'s are two hops away from each other and all the \(\overline{S_{i}}\)'s are two hops away from each other.
\begin{figure}
\begin{center}
\includegraphics[width=.7\linewidth]{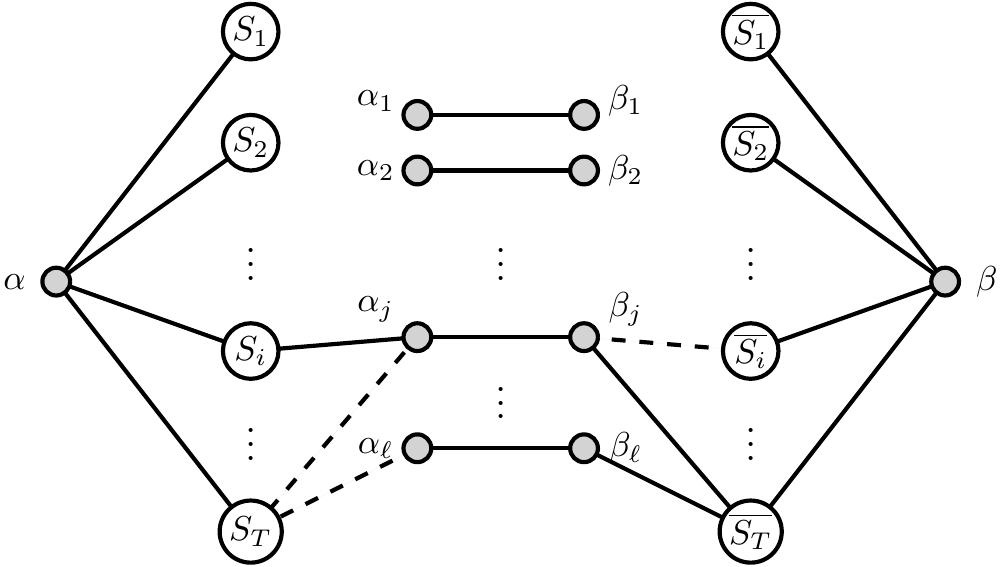}
\end{center}
\caption{\label[figure]{fig:MDS-weighted-bit-gadget} This figure shows the new set gadget we construct inspired by~\cite{BachrachCDELP19}. The lines between $\alpha_j$'s and the $S_i$'s and between
	$\beta_j$'s and $\overline{S_i}$'s indicate membership; bold lines indicate membership and the dashed lines indicate non-membership (e.g., $j \in S_i$, $j \not\in S_T$). Some membership edges are ommitted for clarity.}
\end{figure}
We get the following separation property.
\begin{lemma}
  \label[lemma]{lem:bit-gadget-domination-gap}
The graph \(G^2_{MDS}\) has an MDS of weight \(2\). Moreover, any dominating set of \(G_{MDS}\) that does not pick both \(S_{j}, \overline{S_{j}}\) for some \(j \in \{1, \dots, T\}\) will have size at least \(r\).
\end{lemma}
\begin{proof}
  For any index \(i \in \{1, \dots, T\}\), the vertices \(S_{i}\) and \(\overline{S_{i}}\) form a dominating set of weight \(2\). Note that all the \(\alpha_{i}\)'s and \(\beta_{i}\)'s are covered because either \(S_{i}\) or \(\overline{S_{i}}\) is at most two hops away from them due to the set membership edges.

  We can assume that vertices having weight \(r\) cannot be included in the dominating because otherwise the lemma is vacuously true. Therefore, our only option is to include the \(S_{i}\)'s, and the \(\overline{S_{i}}\)'s in the dominating set.

  The \(r\)-covering property ensures that if we do not pick both \(S_{j}, \overline{S_{j}}\) for some \(j \in \{1, \dots, T\}\), then we will have to pick at least \(r\) set vertices to cover all the \(\alpha_{i}\)'s and the \(\beta_{i}\)'s. The lemma follows.
\end{proof}

We are now ready to describe our lower bound graph construction in detail.
We show how to replace bit gadgets in the construction of $H_{x, y}$ described earlier by
set gadgets $G_{MDS}$.
This modification leads to an MDS of constant weight in $H^2_{x, y}$ and more importantly
an MDS of weight 6 if $x$ and $y$ are not disjoint and a weight of 7 otherwise.

\begin{figure}
\begin{center}
\includegraphics[width=.9\linewidth]{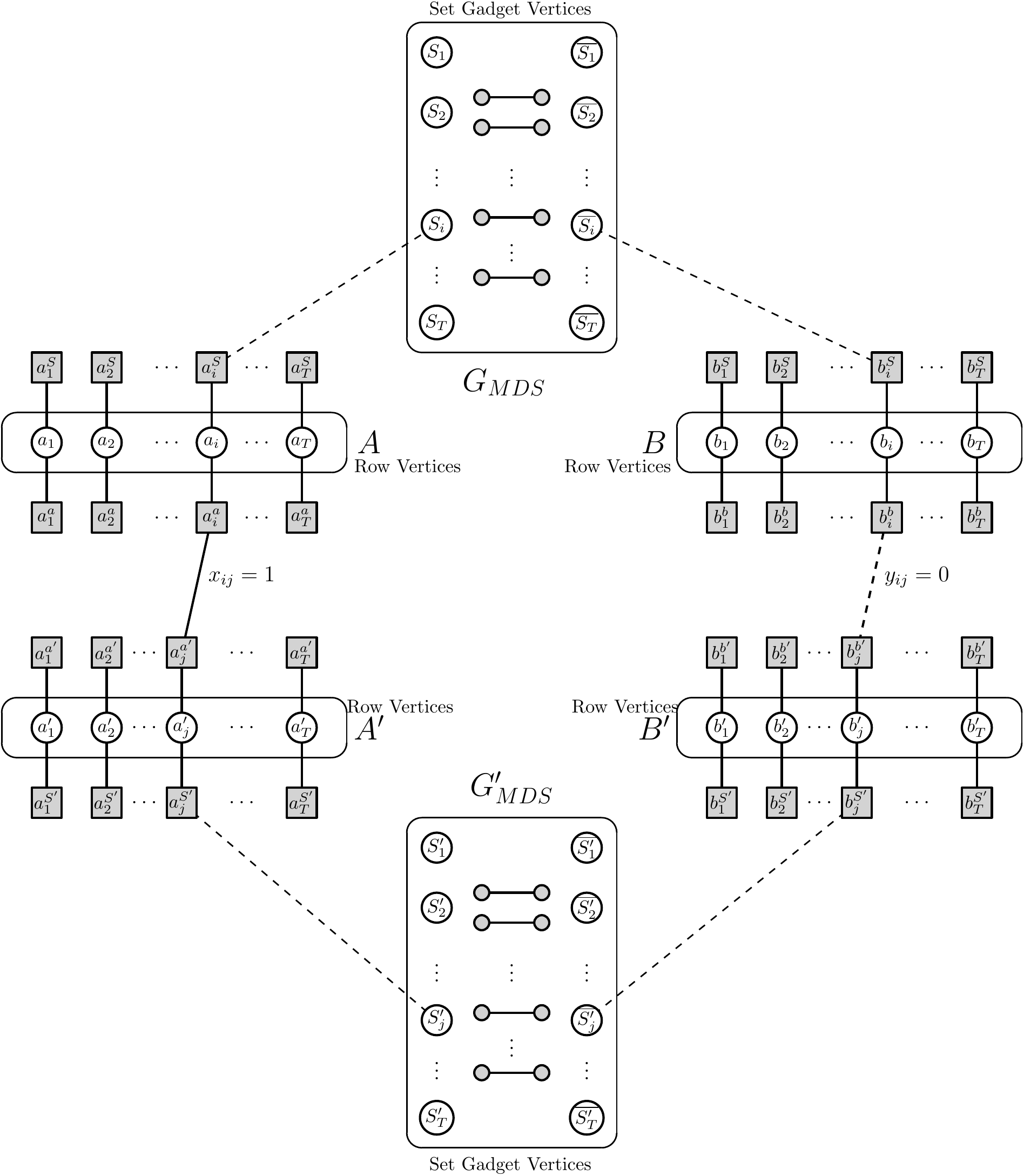}
\end{center}
\caption{\label[figure]{fig:MDS-approximate-lb-graph}This figure shows the lower bound graph \(G_{x,y}\) that shows a quadratic lower bound for computing a \(c\)-approximate solution to MDS for \(c < 7/6\) in the \CONGEST\ model. The dashed edge between a square vertex and a circular vertex \(S_{i}\) means that the square vertex is connected to all \(T-1\) vertices except \(S_{i}\). All square vertices on Alice's and Bob's side are merged path gadgets.}
\end{figure}
\textbf{Fixed Graph Construction:} Our lower bound graph \(H_{x,y}\) consists of four sets of row vertices \(A, A', B, B'\) each of which has \(T\) vertices. There are two copies of the set gadget described above: one connected to \(A, B\), denoted by \(G_{MDS}\), and the other connected to \(A', B'\), denoted by \(G'_{MDS}\). For each vertex \(v \in G_{MDS}\), the corresponding vertex in \(G'_{MDS}\) is named \(v'\).

Each vertex \(a_{i} \in A\), has two shared path gadgets \(A^{a}_{i}\) and \(A^{S}_{i}\), both of which are connected to \(a_{i}\). And similarly, each vertex \(a'_{i} \in A'\) has two shared path gadgets \(A^{a'}_{i}\) and \(A^{S'}_{i}\) which are connected to \(a'_{i}\).

We merge the shared path gadgets \(A^{S}_{i}\), \(A^{a}_{i}\), \(A^{S'}_{i}\), and \(A^{a'}_{i}\) for all \(1 \le i \le T\) to form the merged path gadget \(A_{*}\) having common vertices \(A_{*}[3], A_{*}[4], A_{*}[5]\). The vertex \(A_{*}[3]\) has weight \(0\). The sharing reduces the total number of vertices in the graph, and the merging reduces the weight of the minimum dominating set. Note that \(H_{x,y}\) does not have any dangling path gadgets.

Each \(a^{S}_{i}\) is connected to \(S_{j} \in G_{MDS}\) if \(i \neq j\), and similarly each \(a^{S'}_{i}\) is connected to \(S'_{j} \in G'_{MDS}\) if \(i \neq j\). All vertices on Alice's side except \(\{\alpha_{i}, \alpha'_{i}\}_{i=1}^{\ell}\), \(\alpha\), \(\alpha'\), and \(A_{*}[3]\) have weight 1.
The construction of Bob's side is symmetric, so we do not describe it in this proof.

Alice hosts the vertices in \(A, A'\), the vertices in the merged gadget \(A_{*}\), and the ``left side'' of the set gadgets \(G_{MDS}\) and \(G'_{MDS}\). More formally, the set gadget vertices Alice hosts are \(\alpha\), \(\alpha'\), \(\{S_{i}, S'_{i}\}_{i=1}^{T}\), and \(\{\alpha_{i}, \alpha'_{i}\}_{i=1}^{\ell}\). Bob hosts the rest of the vertices. 

\textbf{Constructing \(H_{x,y}\) given inputs \(x, y \in \{0, 1\}^{T^{2}}\):} We index the strings \(x, y\) by \((i, j) \in T \times T\). We add an edge between \(A^{a}_{i}[1]\) and \(A^{a'}_{j}[1]\) iff \(x_{ij} = 1\) and similarly we add an edge between \(B^{b}_{i}[1]\) and \(B^{b'}_{j}[1]\) iff \(y_{ij} = 1\).

If \(x_{ij} = 1\) then the vertices \(A^{a'}_{j}[1], A^{a}_{i}[1]\) have edges to \(a_{i}\) and \(a'_{j}\) in \(H^{2}_{x,y}\) and if \(x_{ij} = 0\) then no vertex in \(H^{2}_{x,y}\) has an edge to both \(a_{i}\) and \(a'_{j}\). Similarly, if \(y_{ij} = 1\) then the vertices \(B^{b'}_{j}[1], B^{b}_{i}[1]\) have edges to \(b_{i}\) and \(b'_{j}\) in \(H^{2}_{x,y}\) and if \(y_{ij} = 0\) then no vertex in \(H^{2}_{x,y}\) has an edge to both \(b_{i}\) and \(b'_{j}\)

\begin{lemma}\label[lemma]{lem:weighted-mds-disjointness-reduction}
  If \(DISJ_{T^{2}}(x,y) = \mathrm{false}\) then \(H^{2}_{x,y}\) has an MDS of weight \(6\), otherwise any dominating set of \(H^{2}_{x,y}\) has weight at least \(7\).
\end{lemma}
\begin{proof}

  Note that we can assume that \(A_{*}[3]\) and \(B_{*}[3]\) are in the dominating set because their weight is zero (and also due to \Cref{lem:merged-gadget-vertex-in-mds}). The vertex \(A_{*}[3]\) covers all the vertices \(\{A_{*}[4]\), \(A_{*}[5]\), \(A^{a}_{i}[1,2]\), \(A^{S}_{i}[1,2]\), \(A^{S'}_{i}[1,2]\), \(A^{a'}_{i}[1,2]\}_{i=1}^{T}\) and the vertex \(B_{*}[3]\) covers all the vertices \(\{B_{*}[4]\), \(B_{*}[5]\), \(B^{b}_{i}[1,2]\), \(B^{\overline{S}}_{i}[1,2]\), \(B^{\overline{S'}}_{i}[1,2]\), \(B^{b'}_{i}[1,2]\}_{i=1}^{T}\) without increasing the weight of the dominating set.

  If \(DISJ_{T^{2}}(x,y) = \mathrm{false}\) then there is an index \((i,j)\) such that \(x_{ij} = y_{ij} = 1\). Therefore, we add \(A^{a}_{i}[1], S_{i}, B^{b}_{i}[1], \overline{S_{i}}, S'_{j}, \overline{S'_{j}}\) to the MDS incurring a total cost of \(6\). The vertices \(S_{i}\) and \(\overline{S_{i}}\) together cover all the vertices in the set gadget \(G_{MDS}\) along with all row vertices in \(A,B\) except \(a_{i}\) and \(b_{i}\), whereas the vertices \(S'_{j}\) and \(\overline{S'_{j}}\) cover all vertices in the set gadget \(G'_{MDS}\), along with all row vertices in \(A',B'\) except \(a'_{j}\) and \(b'_{j}\). Since \(x_{ij} = y_{ij} = 1\), the vertex \(A^{a}_{i}[1]\) covers both \(a_{i}\) and \(a'_{j}\), and the vertex \(B^{b}_{i}[1]\) covers both \(b_{i}\) and \(b'_{j}\). This means all the vertices in \(H^{2}_{x,y}\) are dominated by a set of weight \(6\) and hence the MDS of \(H^{2}_{x,y}\) has weight at most \(6\) when \(DISJ_{T^{2}}(x,y) = \mathrm{false}\).

  Now we look at the case when \(DISJ_{T^{2}}(x,y) = \mathrm{true}\). In this case, we assume we cannot pick vertex of weight \(r\) in the dominating set because if we do, we immediately get a dominating set of weight at least \(r\). And then the lemma follows because \(r\) was set to some arbitrarily large constant.

  In order to cover all the set gadget vertices in \(G_{MDS}\) and \(G'_{MDS}\), we need to pick \(S_{i}, \overline{S_{i}}, S'_{j}, \overline{S'_{j}}\) for some \(i, j \in \{1, \dots, T\}\). Otherwise we incur a cost of at least \(r\) by \Cref{lem:bit-gadget-domination-gap} for covering vertices in the two set gadgets.

  Now, the only vertices that are left uncovered are \(a_{i}, a'_{j}, b_{i}\), and \(b'_{j}\). Note that since \(DISJ_{T^{2}}(x,y) = \mathrm{true}\), there is no \((i,j)\) such that both \(x_{ij}\) and \(y_{ij}\) are \(1\). Without loss of generality assume \(x_{ij}=0\), therefore there is no vertex in \(H^{2}_{x,y}\) that has an edge to both \(a_{i}\) and \(a'_{j}\). Therefore, we need to pick at least \(2\) vertices in \(H^{2}_{x,y}\) to cover \(a_{i}\) and \(a'_{j}\). And neither of these two vertices will have an edge to \(b_{i}\) and \(b'_{j}\) so we need to pick at least \(3\) vertices to cover all the four vertices. The only vertices with weight less than \(1\) are \(A_{*}[3]\) and \(B_{*}[3]\) which don't cover any of these four vertices. Therefore, every dominating set has to have weight at least \(7\).
\end{proof}

\begin{proof}[Proof of \Cref{thm:approx-MWDS-lower-bound-main-theorem}]
Let \(V_{A} = A \cup A' \cup S \cup S' \cup \{\alpha_{i} \mid 1 \le i \le \ell\} \cup A_{*}\) and \(V_{B} = V \setminus V_{A}\). With these definitions of \(V_{A}\) and \(V_{B}\), the size of the cut \((V_{A},V_{B})\) is at most \(O(\ell) = O(\log T)\). Let \(P\) be the predicate that a graph has a minimum dominating set of weight at least \(7\). \Cref{lem:weighted-mds-disjointness-reduction} implies that \(H^{2}_{x,y}\) is a family of lower bound graphs with respect to the function \(DISJ\) and the predicate \(P\).

Therefore, \Cref{thm: general lb framework} gives an \(\tilde{\Omega}(T^2)\) lower bound for the problem of distinguishing between the case when a graph with \(O(T)\) vertices has a dominating set of weight at least \(7\) and the case when it has a dominating set of weight at most \(6\). This gives a lower bound for approximation factor \(c < 7/6\) which completes the proof of \Cref{thm:approx-MWDS-lower-bound-main-theorem}.
\end{proof}

\subsection{Quadratic Lower bound for $O(1)$-approximate $G^2$-MDS}\label{app:ApproximateMDSLB-unweighted}
The previous lower bound used weights in order to simplify the construction and proofs. In this section, we provide some modifications to get the same lower bound for unweighted MDS. In particular, we will prove the following theorem.

\begin{theorem}
  \label[theorem]{thm:approx-MDS-lower-bound-main-theorem}
  Any distributed algorithm in the \CONGEST\ model which, given an input graph \(G\), produces a \(c\)-approximate solution to the minimum unweighted dominating set problem on \(G^{2}\) for \(c < 9/8\) requires \(\tilde{\Omega}(n^{2})\) rounds.
\end{theorem}

The only modification we need is to the set gadgets \(G_{MDS}\) and \(G'_{MDS}\). In order to do this, we remove the vertices \(\alpha, \beta, \alpha', \beta'\). We connect each \(S_{i}\) to a new vertex \(q_{i}\), each \(S'_{i}\) to a new vertex \(q'_{i}\), each \(\overline{S_{i}}\) to a new vertex \(\overline{q_{i}}\), and each \(\overline{S'_{i}}\) to a new vertex \(\overline{q'_{i}}\). The vertices \(q_{i}, q'_{i}\) are connected to \(A_{*}[3]\), and the vertices \(\overline{q_{i}}, \overline{q'_{i}}\) are connected to \(B_{*}[3]\) for each \(1 \le i \le T\). Since the merged path gadgets in \(H^{2}_{x,y}\) come from only shared path gadgets, we can show the following variant of \Cref{lem:no-other-dangling-gadget-in-mds}. Since the proof is similar, we skip it.

\begin{lemma}\label[lemma]{lem:no-other-merged-shared-gadget-in-mds}
  We can assume w.l.o.g.~that the vertices \(A_{*}[4], A_{*}[5]\) of the merged path gadget \(A_{*}\), and the vertices \(S[2]\) where \(S \in \{A^{S}_{i}, A^{a}_{i}, A^{S'}_{i}, A^{a'}_{i}\}_{i=1}^{T}\) do not belong to the MDS of \(H^{2}_{x,y}\). A similar statement holds for the merged path gadget \(B_{*}\)
\end{lemma}

Therefore, we can show the following lemma which implies \Cref{thm:approx-MDS-lower-bound-main-theorem}.

\begin{lemma}\label[lemma]{lem:mds-disjointness-reduction}
  If \(DISJ_{T^{2}}(x,y) = \mathrm{false}\) then \(H^{2}_{x,y}\) has an MDS of weight at most \(8\) and otherwise any dominating set of \(H^{2}_{x,y}\) has weight at least \(9\).
\end{lemma}
\begin{proof}
  Note that we can assume that \(A_{*}[3]\) and \(B_{*}[3]\) are in the dominating set due to \Cref{lem:merged-gadget-vertex-in-mds}. The vertex \(A_{*}[3]\) covers all the vertices \(\{A_{*}[4]\), \(A_{*}[5]\), \(A^{a}_{i}[1,2]\), \(A^{S}_{i}[1,2]\), \(A^{S'}_{i}[1,2]\), \(A^{a'}_{i}[1,2]\), \(S_{i}\), \(S'_{i}\), \(q_{i}\), \(q'_{i}\}_{i=1}^{T}\) and the vertex \(B_{*}[3]\) covers all the vertices \(\{B_{*}[4]\), \(B_{*}[5]\), \(B^{b}_{i}[1,2]\), \(B^{\overline{S}}_{i}[1,2]\), \(B^{\overline{S'}}_{i}[1,2]\), \(B^{b'}_{i}[1,2]\), \(\overline{S_{i}}\), \(\overline{S'_{i}}\), \(\overline{q_{i}}\), \(\overline{q'_{i}}\}_{i=1}^{T}\).

  If \(DISJ_{T^{2}}(x,y) = \mathrm{false}\) then there is an index \((i,j)\) such that \(x_{ij} = y_{ij} = 1\). Therefore, we add \(A^{a}_{i}[1], S_{i}, B^{b}_{i}[1], \overline{S_{i}}, S'_{j}, \overline{S'_{j}}\) to the MDS incurring a total cost of \(8\). The vertices \(S_{i}\) and \(\overline{S_{i}}\) together cover all the uncovered vertices in the set gadget \(G_{MDS}\) along with all row vertices in \(A,B\) except \(a_{i}\) and \(b_{i}\), whereas the vertices \(S'_{j}\) and \(\overline{S'_{j}}\) cover all the uncovered vertices in the set gadget \(G'_{MDS}\), along with all row vertices in \(A',B'\) except \(a'_{j}\) and \(b'_{j}\). Since \(x_{ij} = y_{ij} = 1\), the vertex \(A^{a}_{i}[1]\) covers both \(a_{i}\) and \(a'_{j}\), and the vertex \(B^{b}_{i}[1]\) covers both \(b_{i}\) and \(b'_{j}\). This means all the vertices in \(H^{2}_{x,y}\) are dominated by a set of weight \(6\) and hence the MDS of \(H^{2}_{x,y}\) has weight at most \(8\) when \(DISJ_{T^{2}}(x,y) = \mathrm{false}\).

  Now we look at the case when \(DISJ_{T^{2}}(x,y) = \mathrm{true}\). Consider the uncovered set gadget vertices in \(G_{MDS}\) which are \(U = {\{\alpha_{i}, \beta_{i}\}}_{i=1}^{\ell}\) and \(G'_{MDS}\) which are \(U' = {\{\alpha'_{i}, \beta'_{i}\}}_{i=1}^{\ell}\). Notice that by \Cref{lem:no-other-merged-shared-gadget-in-mds}, we can assume that \(U\) can only be covered by vertices in \(G_{MDS}\), and \(U'\) can only be covered by vertices in \(G'_{MDS}\).

  The sets \(U, U'\) can be covered using \(4\) vertices: \(S_{i}, \overline{S_{i}}, S'_{j}, \overline{S'_{j}}\) for some \(i, j \in \{1, \dots, T\}\). We cannot cover all the vertices of \(G_{MDS}\) and \(G'_{MDS}\) using fewer than \(4\) vertices, as it would require using at most one vertex to cover all vertices in either \(G_{MDS}\) or \(G'_{MDS}\). This is not possible since no single vertex covers all vertices \({\{\alpha_{i}, \beta_{i}\}}_{i=1}^{\ell}\) in \(G_{MDS}\), and \({\{\alpha'_{i}, \beta'_{i}\}}_{i=1}^{\ell}\) in \(G'_{MDS}\) by the construction and the \(r\)-covering property.

  Note that there are other ways of covering \(G_{MDS}\), and \(G'_{MDS}\) using exactly \(2\) vertices each. If we pick \(q_{i}\) (or \(\overline{q_{i}}\)), it is better to pick \(S_{i}\) (or \(\overline{S_{i}}\)) since it covers more vertices of \(U\). The \(r\)-covering property guarantees that no single set covers all the \(\ell\) elements. But we can also cover \(U\) by picking a vertex \(S \in \{S_{i}, \overline{S_{i}}_{i=1}^{T}\}\) which covers all but one element \(j \in \{1, \dots, \ell\}\) along with either \(\alpha_{j}\) or \(\beta_{j}\). But this is equivalent to picking \(S\) and \(\overline{S}\) because \(\overline{S}\) will cover both \(\alpha_{j}\) and \(\beta_{j}\), along with many other row vertices.

  Therefore, we can assume w.l.o.g.~that the \(4\) vertices used to cover the set gadgets are \(S_{i}, \overline{S_{i}}, S'_{j}, \overline{S'_{j}}\) for some \(i, j \in \{1, \dots, T\}\).

  Now, the only vertices that are left uncovered are \(a_{i}, a'_{j}, b_{i}\), and \(b'_{j}\). Note that since \(DISJ_{T^{2}}(x,y) = \mathrm{true}\), there is no \((i,j)\) such that both \(x_{ij}\) and \(y_{ij}\) are \(1\). Without loss of generality, assume \(x_{ij}=0\), therefore there is no vertex in \(H^{2}_{x,y}\) that has an edge to both \(a_{i}\) and \(a'_{j}\). Therefore, we need to pick at least \(2\) vertices in \(H^{2}_{x,y}\) to cover \(a_{i}\) and \(a'_{j}\). Neither of these two vertices will have an edge to \(b_{i}\) and \(b'_{j}\) so we need to pick at least \(3\) vertices to cover all the four vertices. Therefore, every dominating set has to have size at least \(9\).
\end{proof}

\begin{proof}[Proof of \Cref{thm:approx-MDS-lower-bound-main-theorem}]
Let \(V_{A} = A \cup A' \cup S \cup S' \cup \{\alpha_{i} \mid 1 \le i \le \ell\} \cup A_{*}\) and \(V_{B} = V \setminus V_{A}\). With these definitions of \(V_{A}\) and \(V_{B}\), the size of the cut \(E(V_{A},V_{B})\) is at most \(O(\ell) = O(\log T)\). Let \(P\) be the predicate that a graph has a minimum dominating set of size at least \(9\). \Cref{lem:mds-disjointness-reduction} implies that \(H^{2}_{x,y}\) is a family of lower bound graphs with respect to the function \(DISJ\) and the predicate \(P\).

Therefore, \Cref{thm: general lb framework} gives an \(\tilde{\Omega}(T^2)\) lower bound for the problem of distinguishing between the case when a graph with \(O(T)\) vertices has a dominating set of size at least \(9\) and the case when it has a dominating set of size at most \(8\). This gives a lower bound for approximation factor \(c < 9/8\) which completes the proof of \Cref{thm:approx-MDS-lower-bound-main-theorem}.
\end{proof}
 
\section{Centralized Hardness Results for $G^2$-MVC and $G^2$-MDS}\label{sec:centHardness}
In the following theorem we show that using a dangling-path gadget as in Theorem~\ref{thm:MWVC-lower-bound-main-theorem} gives that MVC is NP-complete on $G^2$, and a simplified version of the proof of Theorem~\ref{theorem:HtoG} gives that there is no FPTAS for MVC on $G^2$ unless $P=NP$. 
\begin{theorem}\label{thm:MVC-centralized-FPTAS-hard}[No FPTAS for $G^2$-MVC]
Given input graph $G$, solving $G^2$-MVC exactly is NP-complete. 
Moreover, there is no FPTAS for $G^2$-MVC unless $P = NP$, i.e., there is no family of algorithms $\{A_\eps \mid \eps>0\}$ such that algorithm $A_\eps$ runs in time $\poly(n, \frac{1}{\eps})$ and yields a $(1 + \eps)$-approximation for MVC on $G^2$, unless $P = NP$.
\end{theorem}

\begin{proof}
  For the first part of the theorem, we use a reduction from MVC on $G$. Given a graph $G=(V_G,E_G)$, we construct a graph $H=(V_H,E_H)$ by replacing each edge $e \in E_G$ with a dangling path gadget $DP_e$ which is a path on three vertices \(p_{e}^{1}, p_{e}^{2}, p_{e}^{3}\) where \(p_{e}^{1}\) is connected to both the end points of \(e\) (as is defined in the proof of \Cref{theorem:HtoG}). Note that \(V_{H}\) contains the vertices in \(V_{G}\) plus the vertices in the dangling path gadget \(DP_{e}\) for each \(e \in E_{G}\). Therefore, the size of \(H\) is polynomial in the size of \(G\). Now we show that \(G\) has a vertex cover of size \(c\) iff \(H^{2}\) has a vertex cover of size \(c + 2|E_{G}|\).

  For the forward direction, consider a vertex cover \(S_{G}\) of \(G\) having size \(c\). We construct a vertex cover \(S_{H}\) of \(H^{2}\) by taking all the vertices in \(S_{G}\) and adding the two vertices \(p_{e}^{1}, p_{e}^{2}\) in the dangling path gadget \(DP_{e}\) for all \(e \in E_{G}\). The two vertices \(p_{e}^{1}, p_{e}^{2}\) cover all the edges in \(H^{2}\) incident on \(DP_{e}\), and \(S_{G}\) covers all the \(H^{2}\) edges that are incident between two vertices in \(V_{G}\). Therefore, \(S_{H}\) is a valid vertex cover of \(H^{2}\) with \(c+2|E_{G}|\) vertices.

  For the reverse direction, consider a vertex cover \(S_{H}\) of \(H^{2}\) having size \(c'\). \Cref{lem:dangling-rearrangement} implies that an exact MVC $S_H$ for $H^2$ takes all vertices of $DP_e$ except $p_{e}^{3}$ for every dangling path gadget $DP_{e}$ and that the set of nodes in $S_H \cap V_{G}$ is a vertex cover for $G$. This in particular implies \(c' \ge 2|E_{G}|\). Let \(c\) be the number of vertices in \(S_{H} \cap V_{G}\). These vertices have to form a valid vertex cover of \(G\) because the subgraph of \(H^{2}\) induced by \(V_{G}\) is exactly \(G\).

  For the second part of the theorem, we follow a line similar to that of Theorem~\ref{theorem:HtoG}, as follows. Let $ALG$ be a $(1+\eps)$-approximation scheme for MVC on $G^2$ that completes in $\poly(n, \frac{1}{\eps})$ time. We construct the same graph $H$ from $G$ as before by adding a dangling path gadget \(DP_{e}\) for each edge \(e \in E_{G}\). We run $ALG$ on \(H^{2}\) with \(\eps = 1/(3|E_{G}|)\). Note that by the previous argument the size of the minimum vertex cover of \(H^{2}\) is \(c + 2|E_{G}|\) where \(c\) is the size of the minimum vertex cover in \(G\). Therefore, $ALG$ will find a vertex cover of size at most \((1+\eps)(c + 2|E_{G}|) = c + 2|E_{G}| + (c + 2|E_{G}|)/(3|E_{G}|) = c + 2|E_{G}| + \alpha\) where \(\alpha < 1\). Therefore, $ALG$ runs in polynomial time and we can find the MVC of \(G\) by taking the solution returned by $ALG$ and taking all the corresponding \(V_{G}\) vertices in the cover. This contradicts the $NP$-hardness of vertex cover in \(G\) (assuming \(P \neq NP\)).
\end{proof}

We also show that one cannot efficiently compute good approximations of $G^2$-MDS unless one obtains a major breakthrough result. 
\begin{theorem}[No better-than-$\ln n$-approximation for MDS on $G^2$]
\label{thm:MDS-centralized-ln-n-hard}
Given input graph $G$, solving MDS exactly on $G^2$ is NP-complete. 
Moreover, if there is some $\eps > 0$ such that a polynomial-time algorithm can solve MDS on $G^2$ to within an approximation factor of $(1 - \eps) \ln n$, then $NP \subseteq DTIME\left(n^{O(\log \log n)}\right)$.
\end{theorem}
\begin{proof}
  We prove this by showing a polynomial time reduction from MDS in \(G\) to MDS in \(G^{2}\). The theorem follows by the hardness of approximation result for \(G\)~\cite{FiegeJACM1998}. The reduction is that for each edge \(e\) in \(G=(V_G,E_G)\), we add a dangling path gadget \(DP_{e}\) and merge all the dangling path gadgets in \(G\) to form the merged path gadget \(DP_{E}\). We call this new graph \(H=(V_H,E_H)\). Note that \(H\) has \(O(m)\) vertices where \(m = |E_{G}|\) which is polynomial in the size of \(G\). The vertices in \(V_{H}\) can be partitioned into two sets, namely the vertices in \(DP_{E}\) and the vertices corresponding to \(V_{G}\), which we call the \(G\)-vertices. Now we show that the size of the MDS in \(H^{2}\) is \(W + 1\), iff the size of the MDS in \(G\) is \(W\).

  We start with the forward direction, let \(S_{H}\) be an MDS of \(H^{2}\) having size \(W+1\). By~\Cref{lem:merged-gadget-vertex-in-mds} we know that \(DP_{E}[3]\) has to belong to \(S_{H}\) (and therefore \(W \ge 0\)) and it covers all the vertices in the merged path gadget \(DP_{E}\). Note that given \(DP_{E}[3]\) is in \(S_{H}\), we can assume that no other vertex of \(DP_{E}\) can belong to \(S_{H}\) by arguing along similar lines as the proof of \Cref{lem:no-other-dangling-gadget-in-mds}. Now \(S_{H}\) needs to cover the \(G\)-vertices in \(H\), without using any vertices in \(DP_{E}\). The subgraph of \(H^{2}\) induced by the \(G\)-vertices is exactly the graph \(G\). So \(S_{H} \setminus \{DP_{E}[3]\}\) must form an MDS of \(G\) because if \(G\) has a dominating set \(S\) of size \(< W\), then it contradicts the optimality of \(S_{H}\) because \(S \cup \{DP_{E}[3]\}\) is a valid dominating set of \(H^{2}\) having size \(< W + 1\).

  To prove the reverse direction, let \(S_{G}\) be an MDS of \(G\) of size \(W\). We construct \(S_{H}\) by taking the \(G\)-vertices in \(H^{2}\) corresponding to \(S_{G}\) along with the vertex \(DP_{E}[3]\). The set \(S_{H}\) has size \(W+1\) and it covers all the vertices in \(H^{2}\) because the vertices corresponding to \(S_{G}\) dominate all the \(G\)-vertices in \(H^{2}\), and \(DP_{E}[3]\) dominates all the vertices in \(DP_{E}\). Therefore, \(S_{H}\) is a dominating set of \(H^{2}\) of size \(W+1\). Note that \(H^{2}\) cannot have a smaller dominating set because then we can use the argument for the forward direction to extract a dominating set of \(G\) of size \(< W\), which contradicts the assumption that \(S_{G}\) is an MDS of \(G\).
\end{proof}
 
\section*{Acknowledgement}
 This project was partially supported by the European Union's Horizon 2020 Research and  Innovation Programme under grant agreement no. 755839 (Keren Censor-Hillel, Yannic Maus).

\bibliographystyle{alpha}
\bibliography{Refs-arxiv}

\end{document}